\documentclass{article}

\usepackage{amsmath, amssymb, amsfonts, enumerate}
\usepackage{microtype}
\usepackage{graphicx, psfrag,epsf}
\usepackage{subfigure, placeins}
\usepackage{booktabs} 
\usepackage{comment}

\newcommand{\rev}[1]{\textcolor{black}{#1}}
\usepackage{hyperref}

\newcommand{\E}{\mathbb{E}}

\newtheorem{condition}{Condition}
\newtheorem{lemma}{Lemma}
\newtheorem{remark}{Remark}

\newtheorem{theorem}{Theorem}
\newtheorem{corollary}{Corollary}
\usepackage{enumitem}

\newcommand{\TnM}{T_n^{(M)}}

\usepackage[accepted]{icml2021}

\icmltitlerunning{Robust Inference for Lasso}

\begin{document}

\twocolumn[
\icmltitle{Robust Inference for High-Dimensional Linear Models \\ via Residual Randomization}



\icmlsetsymbol{equal}{*}

\begin{icmlauthorlist}
\icmlauthor{Y. Samuel Wang}{equal,uc}
\icmlauthor{Si Kai Lee}{equal,uc}
\icmlauthor{Panos Toulis}{uc}
\icmlauthor{Mladen Kolar}{uc}
\end{icmlauthorlist}

\icmlaffiliation{uc}{Booth School of Business, University of Chicago, Chicago, USA}

\icmlcorrespondingauthor{Y. Samuel Wang}{swang24@uchicago.edu}
\icmlcorrespondingauthor{Si Kai Lee}{sikai@uchicago.edu}

\icmlkeywords{high-dimensional inference, debiased Lasso}

\vskip 0.3in
]



\printAffiliationsAndNotice{\icmlEqualContribution} 

\begin{abstract}
We propose a residual randomization procedure designed for robust Lasso-based inference in the high-dimensional setting.
Compared to earlier work that focuses on sub-Gaussian errors, the proposed procedure is designed to work robustly in settings that also include heavy-tailed covariates and errors. 
Moreover, our procedure can be valid under clustered errors, which is important in practice,
but has been largely overlooked by earlier work.
Through extensive simulations, we illustrate our method's wider range of 
applicability as suggested by theory.
In particular, we show that our method outperforms state-of-art methods in challenging, yet more realistic, settings where the distribution of covariates is heavy-tailed or the sample size is small, while it remains competitive in standard, ``well behaved" settings previously studied in the literature.
\end{abstract}

\section{Introduction}
The Lasso~\citep{tibshirani1996Lasso} and its variants are typically used to estimate the coefficients of a linear model in the high-dimensional setting where the number of covariates, $p$, is larger than the number of samples, $n$. The Lasso has been shown to possess many desirable theoretical properties and has proven fruitful in applications across nearly all scientific domains \citep{buhlmann2011statistics}. This widespread use has recently generated much interest in procedures for performing inference using Lasso estimates.

However, for parameters which are zero or nearly zero, the Lasso point estimates may have an irregular distribution, and na\"{i}vely constructing confidence intervals typically results in invalid inference. To overcome these difficulties, various procedures have been proposed.
\citet{wasserman2009selection} and \citet{meinshausen2009pvalues} both used sample splitting procedures to form valid $p$-values. 
\citet{zhang2014debiased}, \citet{vanDeGeer2014asymptotically}, and \citet{javanmard2014confidence} proposed the desparsified and debiased Lasso, which adds a one step correction to the Lasso estimate. The resulting estimate is not sparse, but under certain conditions, it has asymptotically negligible bias. The debiased/desparsified estimate is asymptotically normal and can be used for inference.

An alternative approach for inference in high-dimensional linear models is the
bootstrap. \citet{chatterjee2011bootstrap} proposed a residual bootstrap procedure for the adaptive Lasso \citep{zou2006adaptive}. 
One crucial requirement with this approach is the
``beta-min'' condition, which requires the non-zero parameters to be
large enough in absolute value (i.e., much larger than $n^{-1/2}$).
This condition can be overly restrictive in cases where the primary aim is to test null hypotheses on the significance of regression coefficients of the form $H_0:\beta_j = 0$. To circumvent this problem, \citet{dezure2017simultaneous} and \citet{zhang2017simultaneous} both proposed bootstrap procedures based on the desparsified Lasso~\citep{vanDeGeer2014asymptotically}, which are capable of performing simultaneous inference over a set of parameters in the linear model; i.e., $H_0: \beta_j = 0$ for all $j \in J \subseteq [p] = \{1, \ldots, p\}$. %
\citet{zhang2017simultaneous} proposed bootstrapping the linearized part of the desparsified Lasso with a Gaussian multiplier bootstrap, while \citet{dezure2017simultaneous} proposed using either a wild bootstrap or a residual bootstrap procedure for the entire estimator. 

The above procedures are typically obtained under strong regularity conditions
on covariates and require i.i.d.~sub-Gaussian errors.
Therefore, they may perform poorly in more realistic settings where $n$ is relatively small, the covariates and errors are non-Gaussian and/or heavy-tailed, or have complex structures, such as heterogeneity or clustering.

In contrast, randomization methods~\cite{fisher1937design, pitman1937significance, kempthorne1952design} are non-parametric and typically exact in finite samples,
which makes them robust~\citep[Chapter 15]{lehmann2006testing}. 
Randomization procedures leverage structure in the data for testing and inference ---e.g., permutation tests 
exploit exchangeability through permuting the data---\rev{and rely less on analytical assumptions or asymptotic arguments.}

We show that a robust alternative for inference in high-dimensional linear models is possible through the use of {\em residual randomization}, which was first proposed by \citet{freedman1983nonstochastic, freedman1983significance} as an extension of Fisher's randomization test.
%
Residual randomization builds a test that is exact in the idealized case where the true errors are known, and remains asymptotically valid when using regression residuals in lieu of the true errors.  
Recently,~\citet{toulis2019life} extended the original residual randomization procedure to include complex error structures~(e.g., clustered errors), and showed that the procedure is asymptotically valid and empirically effective in the low-dimensional setting where $p$ is fixed and $n$ grows, 
\rev{including settings with heterogeneity and clustering.} 
\citet{toulis2019life} also proposed an extension to the high-dimensional setting (different than the one we propose), but did not give any theoretical guarantees. 


\subsection{Contributions}
In this paper, we propose a novel residual randomization procedure for conducting hypothesis tests and computing confidence intervals for parameters in a high-dimensional linear model.
In Section~\ref{sec:rr_Lasso}, we define a class of oracle tests 
that are exact for finite samples. 
However, these oracle tests are infeasible, and so we develop an approximate implementation for each test in the class.
Broadly speaking, our procedure selects a specific test from the class for which the approximate but feasible test is close to the oracle. 
Thus, we explicitly prioritize controlling the empirical size of the test.
Our procedure is specifically designed to be robust to small sample sizes 
and non-Gaussianity in both covariates and errors. This is in contrast to previous procedures that prioritize testing power and shorter confidence intervals. 

We show theoretically that our procedure is valid even when the covariates and errors are sub-Weibull as opposed to the sub-Gaussian condition previously required in the literature. We also show that our procedure is sound when the errors have clustered dependence. Indeed, we see empirically that our residual randomization method is comparable to state-of-art methods in ``well behaved'' high-dimensional benchmarks and is superior in more complex settings, e.g., when $n$ is small, the covariates are non-Gaussian, or the errors are heavy-tailed or heterogeneous.

\section{Methodology}
\subsection{Setup}
Throughout we let $\vert \cdot \vert_q$ and $\Vert \cdot \Vert_q$ denote the 
vector $q$-norm and matrix $q$-norm, respectively. For any positive integer $d$, we let $[d] = \{1, \dots, d\}$. For $X \in \mathbb{R}^{n \times p}$, let $X_{i,:}$ denote the $i$th row and $X_{:,v}$ denote the $v$th column. 
We let $e_j$ denote the $j$th standard basis; i.e., 
a vector whose $j$th element is $1$ and all other elements are $0$. 

We assume that the data $Y \in \mathbb{R}^n$ are generated from the linear model 
\begin{equation}
\label{eq:model}
Y = X \beta + \varepsilon,
\end{equation}
where $\beta \in \mathbb{R}^p$ are the linear coefficients, $X \in \mathbb{R}^{n \times p}$ are the covariates, and $\varepsilon \in \mathbb{R}^n$ is the vector of (unobservable) errors. 
Thus, $X_{i,:}$ corresponds to the covariates for the $i$th observation and $X_{:,v}$ corresponds to the $v$th covariate. Let $s$ denote the sparsity of $\beta$ such that $\vert \beta \vert_0 \leq s$.

In contrast to previous work which requires sub-Gaussian covariates and errors, we allow $X_{i,:}$ and $\varepsilon_i$ to follow sub-Weibull($\alpha$) distributions. Sub-Weibull random variables are a class of distributions with tails of the form $\exp(-\vert x\vert^{\alpha})$
\citep{kuchibhotla2018moving,vladimirova2020sub}.
The class of sub-Gaussian and sub-exponential distributions
can be obtained as a subclass by setting $\alpha$ equal to $2$ and $1$, respectively.
However, when $\alpha < 1$, the tails can be heavier than sub-exponential.
Finally, we also assume that $X_{i,:}$ are drawn i.i.d., and $X_{i,:}$ and $\varepsilon_i$ are uncorrelated so that $\E(X_{i,:} \varepsilon_i) = 0$. 

We propose a procedure to test linear hypotheses of the form 
\begin{equation}\label{eq:hypothesisUni}
H_0: a^\top \beta = a_0,
\end{equation}
for some $a \in \mathbb{R}^{p}$ and $a_0 \in \mathbb{R}$; we then invert the test to form confidence intervals. The form in \eqref{eq:hypothesisUni} includes many hypotheses of interest; e.g., setting $a = e_j$ and $a_0 = 0$ implies $H_0: \beta_j = 0$, setting $a = e_j - e_k$ and $a_0 = 0$ implies $H_0: \beta_j = \beta_k$.  In Remark~\ref{rem:simultaneous} we briefly discuss how the procedure can be generalized to tests of the form $H_0: A^\top \beta = a_0$ where $A \in \mathbb{R}^{p \times d}$ and $a_0 \in \mathbb{R}^d$. This would allow for simultaneous inference as in \citet{zhang2017simultaneous} and \citet{dezure2017simultaneous}.

\subsection{Residual Randomization}

Following the framework of~\citet{toulis2019life}, we require two key constructs: (1) A set $\mathcal{G}$ of linear maps $G: \mathbb{R}^n \mapsto \mathbb{R}^n$ such that $G\varepsilon \stackrel{d}{=} \varepsilon$, \rev{conditional on $X$}; and (2) An invariant $t_n: \mathbb{R}^n \rightarrow \mathbb{R}$, where $t_n(\varepsilon) \stackrel{d}{=} t_n(G \varepsilon)$ for all $G \in \mathcal{G}$ and any finite $n$. 
\rev{We emphasize that all the test we propose is conditional on $X$.}

Given these two definitions, if we can find a test statistic, $T_n(Y, X)$, such that $T_n(Y, X) = t_n(\varepsilon)$ under $H_0$, then we can compare the observed value of $T_n$ with $\{t_n(G \varepsilon) : G \in\mathcal{G}\}$ to test $H_0$. 
Indeed, conditioned on $X$ and $\varepsilon$ (or alternatively $X$ and $Y$), the random variable
\begin{equation} \footnotesize \label{eq:pval}
\rev{\pi = \sum_{G\in\mathcal{G}} \mathbb{I}\{ t_n(G \varepsilon) \geq T_n \} / |\mathcal{G}| } ,
\end{equation}
is uniform over the set $\{0, \ldots, |\mathcal{G}|\} / |\mathcal{G}|$. Thus, $\pi$ can be used as a $p$-value and rejecting the null hypothesis when $\pi \leq \pi_0$ yields a hypothesis test with exact size $\pi_0$\footnote{We describe a one sided test, but a two-sided test could be similarly defined.} (save for the discreteness in $\pi$ which can be easily remedied by adding uniform noise to $\pi$ or increasing $|\mathcal{G}|$). \rev{We consider \rev{the following} three invariances  for the distribution of $\varepsilon$.}

\vspace{-.2em}
\textbf{Exchangeability:} If all elements of $\varepsilon$ are exchangeable, then $\mathcal{G}$ could be all $n \times n$ permutation matrices. 
\rev{This includes the standard setting with i.i.d.~errors.}

\vspace{-.2em}
\textbf{Sign Symmetry:} If $\varepsilon_i \stackrel{d}{=} -\varepsilon_i$ for all $i \in [n]$, then $\mathcal{G}$ could be all $n\times n$ diagonal matrices with $\pm 1$ on the diagonal. This allows for heteroskedastic errors where $\varepsilon_i$ may depend on $X_{i,:}$, but is symmetric around $0$ conditional on $X_{i,:}$.

\vspace{-.2em}
\textbf{Clustered Exchangeability:} \rev{In many cases, 
the data can be partitioned into disjoint clusters, such that exchangeability is only reasonable within a cluster; e.g., to model country-specific effects in users of an online platform; $\mathcal{G}$ should then be the set of within-cluster permutations.
This generalizes exchangeability, but we will keep the two settings distinct.}

In the low-dimensional setting, 
where $n \gg p$, \citet{toulis2019life} considered
tests of the form $a^\top \beta = a_0$ and used
the test statistic $T_n = \sqrt{n}(a^\top\hat{\beta} - a_0)$, 
where $\hat{\beta}$ is the least squares estimate of $\beta$, and set the invariant $t_n(u) = \sqrt{n} a^\top (X^\top X)^{-1} X^\top u$.
Since $\hat{\beta} = \beta + (X^\top X)^{-1} X^\top \varepsilon$, 
under the null hypothesis that $a^\top \beta = a_0$, 
\begin{equation}\label{eq:testStat}\footnotesize
    T_n = \sqrt{n}(a^\top\hat \beta - a_0) = \sqrt{n}a^\top(\hat \beta - \beta) = t_n(\varepsilon).
\end{equation} 
Thus, given the true errors $\varepsilon$, we could form an exact $p$-value as in Eq.~\eqref{eq:pval}. In practice, $\varepsilon$ is unknown so one would instead use the residuals, $\hat{\varepsilon} = Y- X\hat \beta$, 
or the restricted residuals, $y - X \hat \beta_\mathrm{r}$, where $\hat \beta_\mathrm{r}$ is the restricted OLS estimates under $H_0$.
In this case, the test is approximate, but as shown in \citet{toulis2019life}, it can attain the correct size asymptotically.

\begin{remark}
The robustness properties of residual randomization can be deduced from~\eqref{eq:pval}. 
Specifically, the performance of the test does not depend on the distribution of the errors. 
Similarly, regularity conditions on $X$ are not needed because the test is conditioned on $X$. 
Finally, the decision of the test remains invariant to monotone transformations of $t_n$, and so the test remains robust to rescaling of the data.
%
In our experiments, we demonstrate that the robustness properties of the exact test (using the true errors) are also inherited by the approximate procedure using the residuals.
\end{remark}

\subsection{High-Dimensional Residual Randomization}
\label{sec:rr_Lasso}
In the high-dimensional setting, where $p > n$, the OLS estimate $\hat \beta$ is ill-defined and the low-dimensional residual randomization method cannot be directly applied. In this section, we propose adjustments appropriate for the high-dimensional setting and show that this test can be further optimized for robustness.

Instead of the OLS estimator, we use a version of the debiased Lasso \cite{javanmard2014confidence} which corrects for the regularization bias by adding a term proportional to the subgradient of the objective at the Lasso solution $\hat{\beta}^l(\lambda_1)$, where $\lambda_1$ is the penalty parameter. 
When it is obvious, we will simply write $\hat{\beta}^l$ instead of $\hat{\beta}^l(\lambda_1)$.
Appropriate values of $\lambda_1$ are problem dependent, and we give theoretically sufficient choices in Section~\ref{sec:theory}.

Specifically, for some $M \in \mathbb{R}^{p \times p}$, we use the estimator
\begin{equation}\footnotesize
\label{eq:dbl}
\hat \beta^{d, M} = \hat \beta^{l} + \frac{1}{n} M X^\top (Y - X \hat \beta^{l}).
\end{equation}
To form the high-dimensional test statistic, 
we replace $\hat{\beta}$ in \eqref{eq:testStat} with $\hat{\beta}^{d,M}$: 
\begin{equation}\footnotesize
\label{eq:T_n}
    \TnM = \sqrt{n} (a^\top \hat{\beta}^{d, M} - a_0).
\end{equation}
Setting $S=\frac{1}{n}X^\top X$, this yields
\begin{equation} \footnotesize
\begin{aligned}\label{eq:ex_T_n}
    \TnM
    &= \sqrt{n} a^\top (I - MS)(\hat \beta^{l} - \beta) + \frac{1}{\sqrt{n}} a^\top M X^\top \varepsilon\\
    &\quad + \sqrt{n} (a^\top \beta - a_0).
\end{aligned}
\end{equation}
For some fixed $\hat \beta^{l}$
and any $M \in \mathbb{R}^{p \times p}$, comparing $T_n^{(M)}$ to
\begin{equation}\footnotesize
\label{eq:oracleDist}
   t^{(M)}(G\varepsilon) = \sqrt{n} a^\top (I - MS)(\hat \beta^{l} - \beta) + \frac{1}{\sqrt{n}} a^\top M X^\top G\varepsilon
\end{equation}
for all $G \in \mathcal{G}$ would yield an exact test as described in \eqref{eq:pval} under the null,
since $a^\top \beta - a_0= 0$.
When the null hypothesis does not hold such that $a^\top \beta = a_{1} \neq a_0$, under weak conditions, $\max_G \frac{1}{\sqrt{n}} a^\top M X^\top G\varepsilon$ will be bounded at a rate of $\log(pn)$. However, $T_n^{(M)}$ contains an additional $\sqrt{n} (a^\top \beta - a_0) = \sqrt{n}(a_1 - a_0)$ term which will
grow as $O(\sqrt{n})$ leading to rejection of the null hypothesis. 

This procedure could be applied for any $M$.\footnote{
If $a^\top M X^\top = 0$, then the distribution over $G$ would be a point mass. Nonetheless some randomization procedure for breaking ties could be used to maintain exact size. In \eqref{eq:selectM}, letting $\lambda <1$  will prevent such an $M$ from being feasible.} Thus, in contrast to the low-dimensional setting, in the high-dimensional setting, we have not just a single test, but a class of tests indexed by $M$, all of which are exact under the null hypothesis. For any $M$, we call this the \emph{oracle randomization distribution}. Since all oracle tests are exact, a good rule of thumb would be to select the matrix $M$,
which gives the test with the most power, or alternatively, which---when inverted---yields the shortest confidence intervals. In some sense, this is the motivation behind \citet{zhang2017simultaneous} and \citet{dezure2017simultaneous} setting $M$ as an estimate of $\Sigma^{-1}$, albeit for a bootstrap procedure which is not exact. 
As suggested by the Gauss-Markov theorem, this should asymptotically give the shortest confidence intervals.  

However, since we do not have access to $\hat \beta^{l} - \beta$ or $\varepsilon$, we cannot directly use any of these oracle tests. Thus, we use the following invariant with $\hat{\varepsilon} = Y - X\hat \beta^{l} = \varepsilon + X(\beta - \hat \beta^{l})$ being the residuals from the Lasso regression:  
\begin{align}
\footnotesize
\begin{split}
\label{eq:residRandomDist}
    \hat{t}^{(M)}(G \hat{\varepsilon}) &= \frac{ a^\top MX^\top  G \hat{\varepsilon}}{\sqrt{n}} = \frac{a^\top MX^\top  G (\varepsilon + X(\beta - \hat \beta^{l}))}{\sqrt{n}} \\
    &= \frac{1}{\sqrt{n}} a^\top MX^\top GX(\beta - \hat \beta^{l}) + \frac{1}{\sqrt{n}} a^\top MX^\top  G\varepsilon .
\end{split}
\end{align}
For any $M$, we refer to the resulting distribution---when selecting $G$ uniformly from $\mathcal{G}$---as the \emph{attainable randomization distribution},
because it only involves quantities which we can directly access and compute. The oracle \eqref{eq:oracleDist} and attainable \eqref{eq:residRandomDist} distributions share the same second term,
but differ in their first terms $(1 / \sqrt{n}) a^\top MX^\top GX(\beta - \hat \beta^{l})$ and $\sqrt{n} a^\top (I - MS)(\hat \beta^{l} - \beta)$. Thus, each attainable distribution no longer retains the exact size that its corresponding oracle test enjoys. In particular, selecting an attainable test based on minimizing its oracle's confidence interval length may result in poor finite sample performance. 

Instead of prioritizing short confidence intervals, we prioritize the correct size of the test by selecting $M$ to minimize the distance between the attainable test and its corresponding oracle test.
Indeed, in Section~\ref{sec:theory} we show that the Wasserstein-1 distance between the oracle and attainable distributions is upper bounded by 
\begin{equation}\label{eq:firstUpperBound}
    \footnotesize
    \sqrt{n} \vert \hat \beta^{l} - \beta \vert_1\left(\vert a^\top(I - M S)\vert_\infty + \frac{\sum_{G} \left\vert a^\top M X^\top G X /n \right\vert_\infty}{\vert \mathcal{G} \vert}  \right).
\end{equation}

Roughly speaking, $\vert a^\top(I - M S)\vert_\infty$ regulates the difference between the means of the attainable and oracle distributions, while $ \sum_G \vert a^\top M X^\top G X / n \vert_\infty / \vert \mathcal{G}\vert $ is the cost of using residuals instead of the true errors. Intuitively, one would expect (and we confirm empirically) that prioritizing the minimization of $\vert a^\top(I - M S) \vert_\infty$ over $\sum_G \vert a^\top M X^\top G X / n \vert_\infty / \vert \mathcal{G}\vert $ has a larger effect on the accuracy of the attainable test's $p$-value. Towards this end, we upweight the first term  by setting $\delta \geq 1$ and select $M$ which minimizes:
\begin{equation}\footnotesize
\begin{aligned}
\label{eq:getLambda}
 \delta \vert a^\top (I - MS)  \vert_\infty + \frac{\sum_{G} \left\vert a^\top M X^\top G X /n\right\vert_\infty  }{\vert \mathcal{G} \vert}.
\end{aligned}
\end{equation}

While \eqref{eq:getLambda} can be solved using 
a linear program solver, 
for computational convenience,
instead of directly optimizing~\eqref{eq:getLambda} with respect to $M$, we instead solve 
\begin{equation}\footnotesize
\label{eq:minOverLambda}
\begin{aligned}
\min_{\lambda \in [0, 1)} \delta \vert a^\top (I - M_\lambda S)  \vert_\infty + \frac{\sum_{G} \left \vert a^\top M_\lambda \right \vert_1 \left \vert X^\top G X/ n  \right\vert_\infty}{\vert \mathcal{G} \vert} ,
\end{aligned}
\end{equation}
where
\begin{equation}\footnotesize
\begin{aligned}\label{eq:selectM}
M_\lambda &= \arg\min_M  \; \vert a^\top M \vert_1\\
&\quad \text{s.t.} \,  \left\vert a^\top (I - MS)\right \vert_\infty \leq \lambda.
\end{aligned}
\end{equation}
The problem in \eqref{eq:selectM} is the CLIME~\cite{cai2011clime} problem, and we solve it using the \texttt{fastclime} package~\citep{fastClime}. 
In Section~\ref{sec:theory} we show that, for any $\delta \geq 1$, 
using $M_{\lambda^\star}$, where $\lambda^\star$ is a minimizer of~\eqref{eq:minOverLambda}, ensures that the 
selected attainable and oracle distributions converge.

To select $M$, \citet{javanmard2014confidence} solve a problem with the same constraint as \eqref{eq:selectM}, but instead minimize the  $M_{i, :}^\top S M_{i, :}$ for all $i \in [p]$, which---similar to \citet{zhang2017simultaneous} and \citet{dezure2017simultaneous}---prioritizes shorter confidence intervals. When it is sparse, the inverse covariance of $X_{i,:}$ can be consistently estimated by solving \eqref{eq:selectM}~\cite{cai2011clime}. In that case the residual randomization procedure should still produce asymptotically efficient confidence intervals and would be asymptotically equivalent to the other procedures. However, in finite samples, we see empirical improvements in robustness. We detail our procedure in Algorithm~\ref{alg:procedure}, which produces a $p$-value for testing the null hypothesis that $a^\top \beta = a_0$. 

\begin{remark}\label{rem:simultaneous}
Thus far, we have assumed a 1-dimensional hypothesis test. However, similar to \citet{zhang2017simultaneous} and \citet{dezure2017simultaneous}, this can generalized be to testing several null-hypotheses simultaneously. In particular, one might instead use
\begin{equation}\footnotesize
\label{eq:T_nMulti}
    T_n = \vert \sqrt{n} (A^\top \hat{\beta}^{d} - a_0) \vert_\infty
\end{equation}
and the corresponding invariant
\begin{equation}\footnotesize
    t_n(G \hat{\varepsilon}) = \left\vert \frac{1}{\sqrt{n}} A^\top MX^\top  G \hat{\varepsilon} \right\vert_\infty.
\end{equation}
We focus on the $1$-dimensional case for expositional clarity. 
\end{remark}


\begin{algorithm}[t]
    		\caption{\label{alg:procedure}Test $a^\top \beta = a_0$} 
		\begin{algorithmic}
            \REQUIRE $Y$, $X$, $a^\top$, $a_0$, $\mathcal{G}$, $\lambda_1$, $\delta$
            \STATE Compute Lasso estimate $\hat \beta^l(\lambda_1)$ and $\hat{\varepsilon} = Y - X\hat \beta^l$
            \STATE Compute $\lambda^\star$ and $M_{\lambda^\star}$ from~\eqref{eq:selectM} and \eqref{eq:minOverLambda}
            \STATE Compute $\hat \beta^{d,M_{\lambda^\star}} = \hat \beta^{l} + \frac{1}{n} M_{\lambda^\star} X^\top (Y - X\hat \beta^{l})$
        \STATE Let $T_n^{(M_{\lambda^\star})} = \sqrt{n} (a^\top \hat \beta^{d,M_{\lambda^\star}} - a_0)$
        \FOR{$G \in \mathcal{G}$}
            \STATE Set $\hat{t}^{(M_{\lambda^\star})}(G \hat{\varepsilon}) = \frac{1}{\sqrt{n}} a^\top M_{\lambda^\star} X^\top G \hat{\varepsilon}$
        \ENDFOR
        \STATE \textbf{Return} $\frac{\sum_G \mathbb{I}\{ \hat{t}^{(M_{\lambda^\star})}(G \hat{\varepsilon}) >  T_n^{(M_{\lambda^\star})}  \}}{\vert \mathcal{G}\vert} $ 
		\end{algorithmic}
\end{algorithm}


\subsection{Confidence Intervals}\label{sec:confInt} 
To form a univariate confidence interval for $\beta_j$, we invert the hypothesis test for $\beta_j = a_0$~\cite{rosenbaum2003exact}. In particular, for $a = e_j$ we can compute the distribution of $t(G \hat{\varepsilon})$ and set $\tau_{\pi_0/2}$ and $\tau_{1 - \pi_0/2}$ to the $\pi_0/2$ and $1 - \pi_0 / 2$ quantiles. Finally, to invert the level $\pi_0$ two-sided test we select $a_0$ such that $\sqrt{n}(\hat \beta^d - a_0) \in [\tau_{\pi_0/2}, \tau_{1 - \pi_0 / 2}]$. 
This is equivalent to the confidence interval for $\beta$:
\begin{equation} \label{eq:confInt}\footnotesize
\left(\hat \beta^d - \frac{\tau_{1 - \pi_0/2}}{\sqrt{n}}, \hat \beta^d - \frac{\tau_{\pi_0/2}}{\sqrt{n}}\right).
\end{equation}
Note that as opposed to the multiplier bootstrap confidence intervals proposed by \citet{zhang2017simultaneous}, the confidence intervals in \eqref{eq:confInt} may be asymmetric because they are produced from the attainable randomization distribution,
which may be asymmetric. This is similar to the procedure proposed by \citet{dezure2017simultaneous}. 

\section{Main Results}\label{sec:theory}
Let $F_{t}(X, \varepsilon)$ denote the oracle randomization distribution of $t_n$ in~\eqref{eq:oracleDist} conditional on $X$ and $\varepsilon$ when $G$ is selected uniformly from $\mathcal{G}$. Let $F_{\hat t}(X, \varepsilon)$ denote the attainable distribution; i.e., 
the distribution of $\hat t_n$ in~\eqref{eq:residRandomDist} when $G$ is chosen uniformly from $\mathcal{G}$. As the notation implies, both $F_{t}(X, \varepsilon)$ and $F_{\hat t}(X, \varepsilon)$ depend on $X, \varepsilon$ and are random distributions with respect to $X, \varepsilon$. We give a finite sample characterization of the Wasserstein-1 distance between these two random distributions under certain assumptions, and we show that the distance goes to $0$ with probability going to $1$. All proofs are given in the supplement.

We first state Lemma~\ref{lem:wasserstein} which, as mentioned in Section~\ref{sec:rr_Lasso} shows that the distance between the oracle and attainable distributions can be decomposed into the estimation error of $\hat \beta^{l}$ as well as two terms which, roughly speaking, regulate the difference in means and variance.  
\begin{lemma}\label{lem:wasserstein}
For any $M \in \mathbb{R}^{p \times p}$, let $d_1\left(F_{t}(X, \varepsilon),  F_{\hat t}(X, \varepsilon)\right)$ denote the Wasserstein-1 distance between the oracle randomization distribution and attainable randomization distributions. Then,
\begin{equation}\footnotesize
\begin{aligned}\label{eq:w1}
    d_1&\left(F_{t}(X, \varepsilon),  F_{\hat t}(X, \varepsilon)\right) \leq\left \vert \hat \beta^{l} - \beta \right \vert_1 \times \\
    &\quad \left[ \left\vert\sqrt{n}a^\top (I - M S)  \right \vert_\infty  + \left \vert a^\top M \right \vert_1 \E_Q \left(\left\vert X^\top  G X  /\sqrt{n} \right\vert_\infty\right) \right].
\end{aligned}
\end{equation}
where $Q$ is the uniform distribution over $G$ in $\mathcal{G}$. 
\end{lemma}

We now provide conditions under which the two terms in \eqref{eq:w1} can be controlled for the $M_{\lambda^\star}$ selected by~\eqref{eq:selectM} and \eqref{eq:minOverLambda}. Condition~\ref{con:covariates} requires that the tails of $X_{i,:}$ and $\varepsilon_i$ be sub-Weibull and bounds certain moments of the observed covariates. 
In Condition~\ref{con:covariates}, $\Vert \varepsilon_i \Vert_{\Psi_\alpha}$ denotes the Orlicz norm of $\varepsilon_i$ with $\Psi_\alpha(x) = \exp(\vert x\vert ^\alpha) - 1$ and $\Vert \tilde X_{i,:} \Vert_{J,\Psi_\alpha} = \sup_{\vert \theta \vert_2 = 1} \Vert \tilde X_{i,:}^\top \theta \Vert_{\Psi_\alpha} $ is the joint Orlicz norm. 
\begin{condition}[Covariates]\label{con:covariates}
Suppose that $X_{i,:} \in\mathbb{R}^p$ are generated i.i.d. with mean $0$ and covariance $\Sigma$. Let $\lambda_{\max}$ and $\lambda_{\min}$ denote the largest and smallest eigenvalues of $\Sigma$. Suppose each element of $X_{i,:}$ is sub-Weibull($\alpha$) and the de-correlated covariates $\tilde X_{i,:}^\top  = \Sigma^{-1/2}X_{i,:}^\top$ are jointly sub-Weibull($\alpha$) with 
\begin{equation}\footnotesize
\max\left(\Vert \tilde X_{i,:}^\top \Vert_{J,\Psi_\alpha}, \max_{v} \Vert X_{i,v} \Vert_{\Psi_\alpha} \right) \leq \kappa.
\end{equation}
We also define
{\scriptsize
\begin{equation}\label{eq:gamma}
\begin{aligned}
    \Gamma =  \max&\left\{\max_{u, v \in [p]^2} \E\left(\left[X_{i,u} X_{i,v}\right]^2\right), \max_{v \in [p]}\E\left(\left[a^\top \Sigma^{-1}X_{i,:}^\top X_{i,v}\right]^2\right), \right.
    \\
    &\qquad \left.  \max_{v \in [p]}\E\left(\left[a^\top \Sigma^{-1}X_{i,:}^\top X_{j,v}\right]^2\right) \right\}.
\end{aligned}
\end{equation}
}
\end{condition}

We assume the high-dimensional regime where $n$ and $p$ both grow and $p$ can be much larger than $n$. Condition~\ref{con:highD} also implicitly restricts the 2-norm of $\Sigma$ and $\Sigma^{-1}$ through $\kappa^\star$. In particular, we require $n\log(n)^{-4/\alpha}\log(pn)^{1 - 4/\alpha}$ to scale linearly with the condition number of $\Sigma$. We will also require $8\sqrt{\Gamma (\log(pn) + 2\log(p))/n} < 1$ so that $\Sigma^{-1}$ is in the feasible set of \eqref{eq:selectM} for some $\lambda< 1$ with probability tending to $1$.

\begin{condition}[Sample Size]\label{con:highD}
Suppose 
\begin{equation}\footnotesize
    \kappa^\star =\kappa^2\max\left(\vert a \vert_2 \frac{ \sqrt{\lambda_{\max}}}{\sqrt{\lambda_{\min}}}, 1 \right)
\end{equation}
and
{\footnotesize \begin{equation}\begin{aligned}
    n > \max&\left\{\frac{4C_\alpha^2 (\kappa^\star)^2 \left[\log(2n)\right]^{4/\alpha}\left[3\log(pn)\right]^{4/\alpha - 1}}{\Gamma},\right.\\
    & \quad\left. \vphantom{\frac{4C_\alpha^2 (\kappa^\star)^2 \left[\log(2n)\right]^{4/\alpha}\left[3\log(pn)\right]^{4/\alpha - 1}}{\Gamma}} 64\Gamma(\log(pn) + 2 \log(p)) \right\}
    \end{aligned}
\end{equation}}
for some constant $C_\alpha$  which only depends on $\alpha$. 
\end{condition}

We now give conditions on the set of group actions, $\mathcal{G}$.

\begin{condition}[Exchangeability]\label{con:exchangeable}
Let $\mathcal{G} \subset \mathcal{G}_p$ where $\mathcal{G}_p$ is the set of all matrices corresponding to a permutation $g$ of $[n]$ such that (i) $[n] = N_1 \cup N_2$ for some $N_1$ and $N_2$ equal-sized disjoint sets,
and (ii) for all $j \in N_1$, $g(j) \in N_2$ and for all $j \in N_2$, $g(j) \in N_1$.
\end{condition}

\begin{condition}[Sign Symmetry]\label{con:symmetry}
Let $\mathcal{G} \subset \mathcal{G}_s$ where $\mathcal{G}_s$ is the set of all diagonal matrices containing only $\pm1$ such that there is an equal number of positive and negative $1$'s. 
\end{condition}

\begin{condition}[Cluster Exchangeability]\label{con:ClustExchangeable}
Suppose there exist $n_c$ disjoint sets $L_k$ with $[n] = \bigcup_k^{n_c} L_k$ and $|L_k| = n / n_c = J$ such that $\{\varepsilon_i\}_{i \in L_k}$ are exchangeable, but may otherwise be dependent. 
That is, $\mathcal{G} \subset \mathcal{G}_c$, where $\mathcal{G}_c$ is the set of all block diagonal matrices where the $G_{L_k, L_k}$ block is a permutation matrix satisfying Condition~\ref{con:exchangeable}. 
\end{condition}

In Conditions~\ref{con:exchangeable},~\ref{con:symmetry}, and \ref{con:ClustExchangeable} we implicitly assume that $n$ is even to simplify the analysis; when $n$ is odd, the last observation may be discarded.
A more subtle requirement is that we do not assume the support of the randomization distribution is over all possible maps in either $\mathcal{G}_p$, $\mathcal{G}_c$, or $\mathcal{G}_s$. These sets grow exponentially in $n$. Instead, we consider the more practical scenario where $\vert \mathcal{G}\vert$ is fixed with respect to $n$. This is similar to using the bootstrap or a Monte Carlo procedure with a number of draws that may be increased when a better approximation is desired, but in general stays fixed with $n$.
Condition~\ref{con:ClustExchangeable} also implicitly requires that $\vert L_k \vert > 2$ since $\vert \mathcal{G}_c\vert = 1$ if all clusters have size $2$ because there would be only $1$ sub-matrix which satisfies Condition~\ref{con:exchangeable} for each $L_k$.   

Given the conditions on the covariates and group actions, we now state Lemma~\ref{lem:permNuisance} and~\ref{lem:sigmaFeasible}, which show that the two terms in Lemma~\ref{lem:wasserstein} can be controlled.
\begin{lemma}\label{lem:permNuisance}
Under Conditions~\ref{con:covariates} and~\ref{con:highD} and either Condition~\ref{con:exchangeable},~ \ref{con:symmetry}, or \ref{con:ClustExchangeable}, we have
\begin{equation}\scriptsize
P\left( \frac{\sum_{G }\vert X^\top G X\vert_\infty}{\vert \mathcal{G}\vert } \geq  8 \sqrt{ \frac{2\Gamma (\log(pn) + 2\log(p))}{n}}\right) \leq  6\vert \mathcal{G}\vert (np)^{-1}.
\end{equation}
\end{lemma}

\begin{lemma}\label{lem:sigmaFeasible}
Under the Conditions~\ref{con:covariates} and \ref{con:highD}, we have
\begin{equation}\label{eq:sigmaFeasible}\scriptsize
P\left(\vert a^\top(I-\Sigma^{-1}S \vert_\infty \geq  8 \sqrt{ \frac{\Gamma (\log(pn) + 2\log(p))}{n}}\right) \leq  3(np)^{-1}.
\end{equation}
Thus, with probability at least $1 - 3(np)^{-1}$ the feasible set of \eqref{eq:selectM} is non-empty with $\lambda = 8 \sqrt{ \frac{\Gamma (\log(pn) + 2\log(p))}{n}}$ and 
\begin{equation}\footnotesize
    \vert a^\top M_\lambda \vert_1 \leq  \vert a^\top \Sigma^{-1} \vert_1.
\end{equation}

\end{lemma}

\begin{corollary}\label{cor:involvesBeta}
Assume the conditions of Lemma~\ref{lem:sigmaFeasible} and Lemma~\ref{lem:permNuisance}. Then with probability greater than $1 - 3(np)^{-1} -  6\vert \mathcal{G}\vert (np)^{-1}$ using $M^\star$ selected from \eqref{eq:selectM} and \eqref{eq:getLambda} yields
\begin{equation}\footnotesize \begin{aligned}
        d_1&\left(F_{t}(X, \varepsilon),  F_{\hat t}(X, \varepsilon)\right) \leq\left \vert \hat \beta^{l} - \beta \right \vert_1 \times \\
    &\quad\left[ 8\left(\delta + \left\vert a^\top \Sigma^{-1} \right \vert_1 \right) \sqrt{2\Gamma (\log(pn) + 2\log(p))} \right]. 
\end{aligned}\end{equation}
\end{corollary}

Corollary~\ref{cor:involvesBeta} implies that using any procedure that can produce an estimate $\hat \beta^{l}$ such that $\vert \hat \beta^{l} - \beta \vert_1 = O_p((\log p)^{-1/2})$ is sufficient for showing that the oracle and attainable distributions converge in Wasserstein distance. For concreteness, we consider two settings and apply existing results on Lasso estimation. However, other estimators (i.e., SCAD or best subset selection) can also be used as long as $\vert \hat \beta - \beta \vert_1$ attains the correct rate. First, we consider Lasso estimates when $\varepsilon_i$ is sub-Weibull($\alpha$) as in \citet{kuchibhotla2018moving} who require some additional assumptions summarized in Condition~\ref{con:lassSubWei} that follows. We also consider the setting of \citet{belloni2016inference} who require sub-Gaussian covariates and errors, but allow for clustered error dependence.

\begin{figure*}[t]
\centering
\includegraphics[scale = .9]{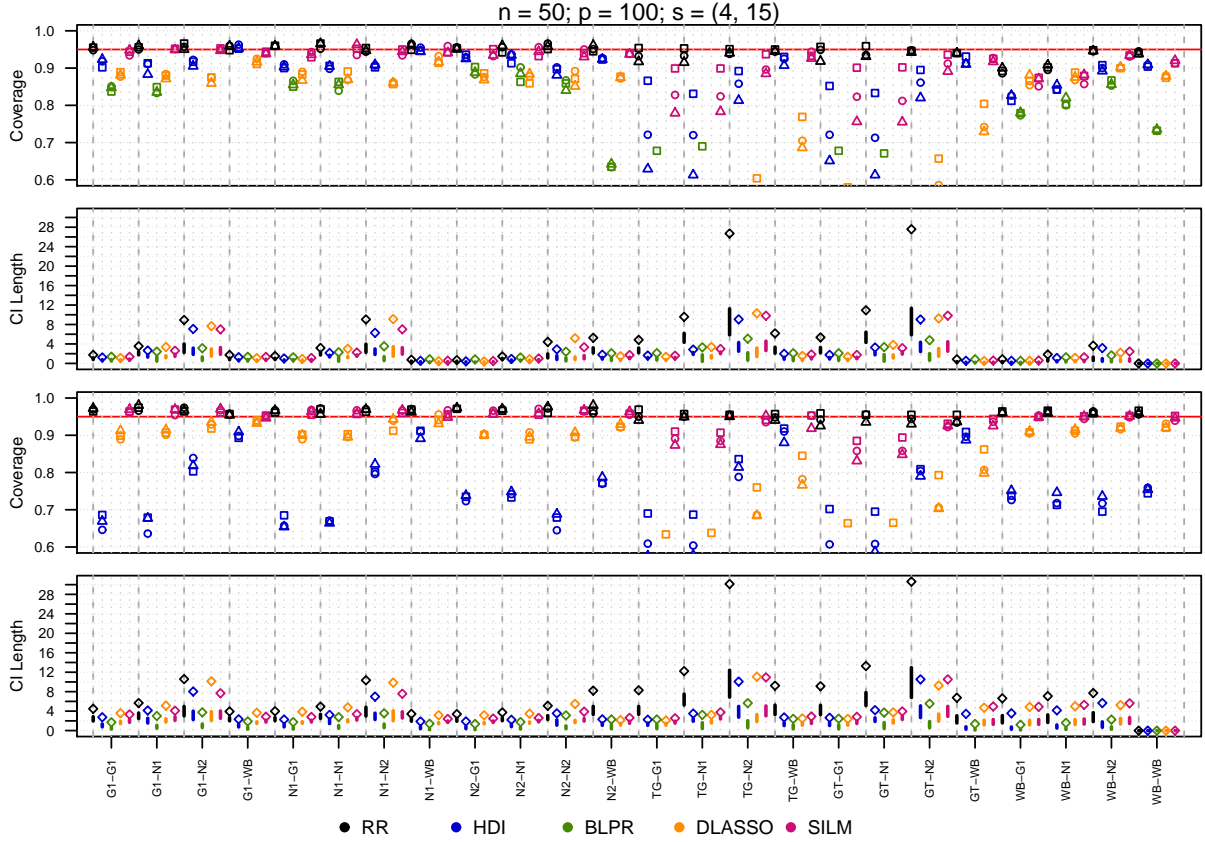}
\vspace{-1em}
\caption{Empirical coverage and confidence interval length for the \emph{active variables} with $n = 50$, $p = 100$, $1000$ replications and \emph{exchangeable errors}. The top two panels are for $s=4$ and the bottom two are for $s = 15$. The first and third panels show empirical coverage rates for each procedure; the sandwich coordinate is denoted by $\Delta$, isolated is $\Box$, and adjacent is $\circ$. In the bottom panel, the line segment spans the $.25$ quantile and $.75$ quantile of the confidence interval lengths and the single point indicates the $.99$ quantile. Instead of showing the quantiles for each coordinate, we instead plot the maximum $.25$ (or $.75$, $.99$) quantile across the sandwich, isolated, and adjacent coordinates. The labels on the horizontal axis indicate a different simulation setting and are coded as ``Covariate - Errors'' where the different covariate and error settings are detailed in the main text. For some settings and procedures, the empirical coverage drops below $.6$ and is not shown.}
\label{fig:small}
\end{figure*}

\begin{condition}[Lasso with sub-Weibull errors]\label{con:lassSubWei}
Suppose $\varepsilon_i$ is sub-Weibull($\alpha$) with $\Vert \varepsilon_i \Vert_{\Psi_\alpha} \leq \kappa$. Suppose that
\begin{equation}\footnotesize
    \lambda_{\min} \geq 54 \min_{1 \leq h \leq p} \left\{ \Xi_{n,h} + \frac{32k\Xi_{n,h} }{h} \right\}
\end{equation}
where \begin{equation}\footnotesize
\begin{aligned}
    \Xi_{n,h} &= 14 \sqrt{2}\sqrt{\frac{\Upsilon_{n,h} h \log(36np/h)}{n}}\\
    &\quad + \frac{C_\alpha \kappa^2 h(\log(2n))^{2/\alpha}(h \log(36np/h))^{2/\alpha}}{n}\\
    \Theta_h &= \left\{\theta \in \mathbb{R}^p : \vert \theta \vert_0 \leq h, \vert \theta \vert_2 \leq 1\right\}\\
    \Upsilon_{n,h} &= \sup_{\theta \in \Theta_h}{\rm var}\left[\left(X_{i,:}^\top \theta\right)^2 \right].
\end{aligned} 
\end{equation}
Furthermore, suppose that the Lasso penalty term $\lambda_1$ is set such that
\begin{equation}\scriptsize
    \lambda_1 = 14 \sqrt{2}\sigma \sqrt{\frac{\log(np)}{n}} + \frac{C_{\alpha/2} \kappa^2(\log(2n))^{2/\alpha}(2\log(np))^{2/\alpha}}{n},
\end{equation}
and in addition to Condition~\ref{con:highD}
\begin{equation}\footnotesize
    n >  \frac{C_{\alpha/2}^2 \kappa^4 (\log(pn))^{8/\alpha -1}}{\sigma^2},
\end{equation}
where $\sigma = \max_{v \in [p]} \text{\rm var}(X_{i,v} \varepsilon_v)$ and $C_{\alpha/2}$ is a constant only depending on $\alpha$.
\end{condition}

\begin{theorem}[Sub-Weibull Errors and Covariates]\label{thm:lassSubWei}Suppose Conditions~\ref{con:covariates}, \ref{con:highD}, and~\ref{con:lassSubWei} hold. Under either Condition~\ref{con:exchangeable} or \ref{con:symmetry}, with probability not less than $1 - \frac{6\vert \mathcal{G} \vert + 3}{np} - \frac{3}{np} - \frac{3}{n}$, 
\begin{equation}\scriptsize
\begin{aligned}
    d_1\left(F_{t}(X, \varepsilon),  F_{\hat t}(X, \varepsilon)\right) &\leq \frac{1360 \left(\delta + \left\vert a^\top \Sigma^{-1} \right \vert_1\right)\sigma s \sqrt{\Gamma}  }{\lambda_{\min}}\frac{\log(np)}{\sqrt{n}}.
    \end{aligned}
\end{equation}
\end{theorem}

Recall in Condition~\ref{con:ClustExchangeable} with clustered errors, $n_c$ denotes the number of clusters, where each cluster has size $J$. Then, Theorem~\ref{thm:clustered}, presented below, combines Corollary~\ref{cor:involvesBeta} with the results of \citet{belloni2016inference} on Lasso estimates under clustered errors. In particular, they propose the Cluster-Lasso procedure and show that its performance depends on a term which measures the within-cluster dependence:
\begin{equation}\scriptsize
    \imath_J = J \min_{1 \leq v \leq p} \E\left(\frac{1}{J} \sum_{j = 1}^J \ddot{X}^2_{ijv} \ddot{\varepsilon}^2_{ij}\right) \bigg/ \E\left(\frac{1}{J} \left[\sum_{j = 1}^J \ddot{X}^2_{ijv} \ddot{\varepsilon}^2_{ij}\right]^2\right),
\end{equation}
where $\ddot{X}_{ijv}$ and  $\ddot{\varepsilon}_{ij}$ denotes the $v$th covariate and $j$th observation in the $i$th cluster which have been adjusted by their respective cluster means. Under complete independence $\imath_J = J$, and the rate in Theorem~\ref{thm:clustered} recovers the rate under independent errors. In the worst case, however,  $\imath_J = 1$, so that each cluster is essentially one observation. It is worth noting that \citet{belloni2016inference} allow cluster dependence in both covariates and errors; however, we allow for dependent errors but still require the covariates to be i.i.d. For completeness, we include the assumptions of \citet{belloni2016inference} in the supplement.

\begin{theorem}[Clustered Errors]\label{thm:clustered}
Suppose Conditions~\ref{con:covariates}, \ref{con:highD}, and \ref{con:ClustExchangeable} hold and $\vert \mathcal{G} \vert = O(1)$. Further assume the conditions of Theorem 1 of \citet{belloni2016inference} and let $\hat \beta^{l}$ be the Cluster-Lasso. Then $d_1\left(F_{t}(X, \varepsilon),  F_{\hat t}(X, \varepsilon)\right)$ is 
\begin{equation}\footnotesize
\begin{aligned}
        O_p\left( \frac{s \left(\delta + \left\vert a^\top \Sigma^{-1} \right \vert_1 \right) \sqrt{\Gamma_p} (\log(pn) + \log(p))}{\sqrt{n_c {\it \imath_J}}}  \right).
\end{aligned}
\end{equation}
\end{theorem}

In Theorems~\ref{thm:lassSubWei} and \ref{thm:clustered}, we explicitly include the term $\vert a^\top \Sigma^{-1} \vert_1$, and for convergence, we require $\vert a^\top \Sigma^{-1} \vert_1 = O(\sqrt{n} /(s\log(np)))$. 

\begin{remark}
Convergence in Wasserstein-1 implies weak convergence, as well as $L_1$ convergence of the quantile and distribution functions. Similar to \citet{bickel1981some}, \citet{bickel1983bootstrapping}, and \citet{lopes2014bootstrap}, we use this to justify the procedure's use for hypothesis tests and confidence intervals. 
\end{remark}

\section{Numerical Experiments}\label{sec:experiments}

\begin{figure*}[t]
\centering
\includegraphics[scale = .9]{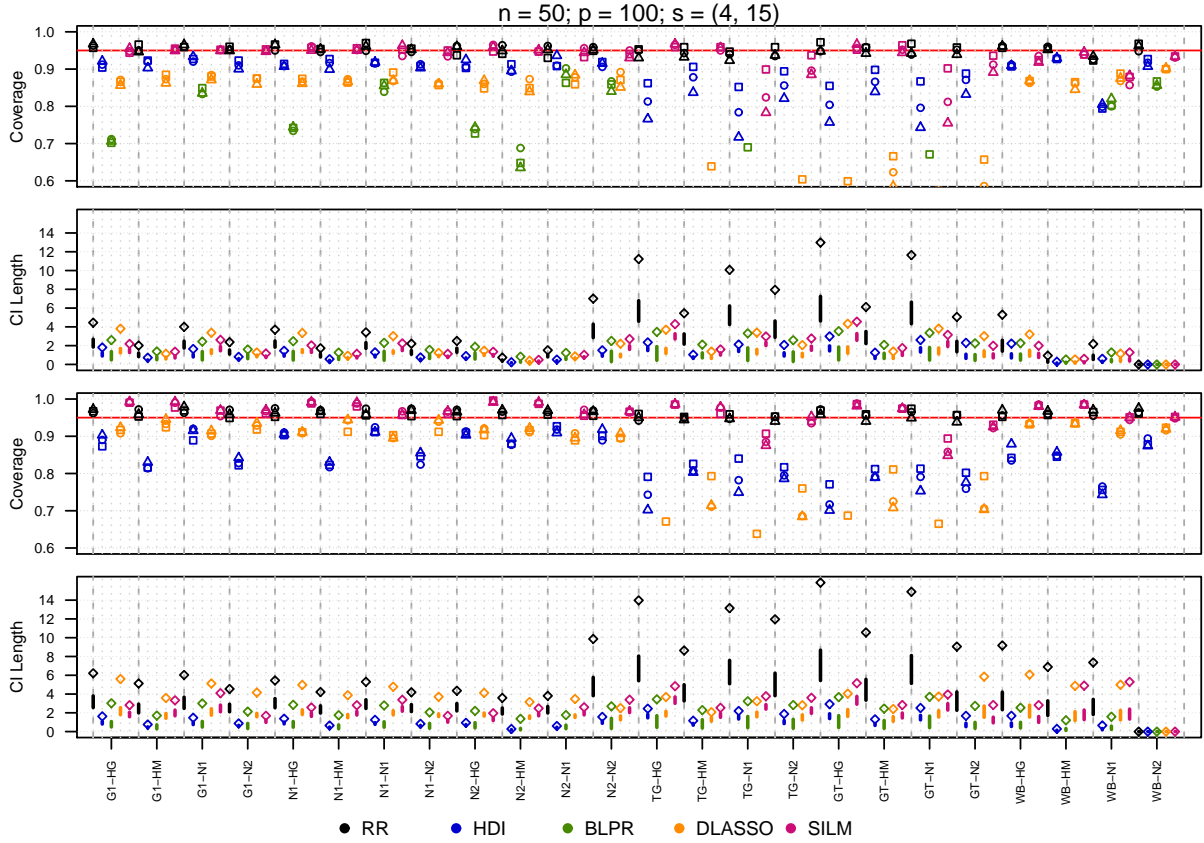}
\vspace{-1em}
\caption{Empirical coverage and confidence interval length for the \emph{active variables} with $n = 50$, $p = 100$, $1000$ replications and \emph{sign symmetric errors}. The top two panels are for $s=4$, the bottom two panels are for $s = 15$. All other elements remain the same as Figure~\ref{fig:small}.}
\label{fig:large}
\end{figure*} 

We compare nominal $95\%$ confidence intervals (CIs) over 1000 trials of \texttt{BLPR}~\citep{liu2017bootstrap}, \texttt{HDI} \citep{dezure2017simultaneous}, \texttt{DLASSO}~\citep{javanmard2014confidence}\footnote{https://web.stanford.edu/~montanar/sslasso/code.html}, \texttt{SILM}~\citep{zhang2017simultaneous} and residual randomization (\texttt{RR}) with $(n = 50, p = 100)$ and $(n = 100, p = 300)$. We slightly modified \texttt{SILM} to output marginal confidence intervals and ensure that the modified variance estimator does divide by $0$ if the support of the estimated $\beta$ is $n$. Additional details are in the supplement, and the code is available at: \url{https://github.com/atechnicolorskye/rrHDI}. 

In each setting, we sample random $X \in \mathbb{R}^{n \times p}$ with rows drawn i.i.d. from either (\textbf{N1}) $N(0, I)$; (\textbf{G1}) $\text{Gamma}(1, 1) - 1$; (\textbf{N2})  $N(\mu ,1)$ with $P(\mu = -2) = P(\mu = 2) = 0.5$; (\textbf{NT}) $N(0, \Sigma)$ for $\Sigma_{ij} = .8^{|i-j|}$; (\textbf{GT}) $\text{Gamma}(\Sigma) - 1$ for $\Sigma_{ij} = .8^{|i-j|}$; or (\textbf{WB}) each element is a centered Weibull with scale $=1$ and shape $=1/2$. 

We sample the errors $\varepsilon \in \mathbb{R}^n$ from (\textbf{N1}) $N(0, 1)$; (\textbf{G1}) $\text{Gamma}(1, 1) - 1$; (\textbf{N2}) $N(\mu ,1)$ with $P(\mu = -2) = P(\mu = 2) = 0.5$; (\textbf{WB}) a centered Weibull with scale $ =1$ and shape $ =1/2$; (\textbf{HN}) normal with $N(0, 2 \Vert X_i \Vert_2^2 / p)$; or (\textbf{HM}) $N(\mu, 2 \Vert X_i \Vert_2^2 / p)$ with $P(\mu = -2) = P(\mu = 2) = 0.5$. The exchangeable settings exclude the heteroskedastic cases of \textbf{HN} and \textbf{HM}; the symmetric settings exclude \textbf{G1} and \textbf{WB}.

For a fair comparison, we have \texttt{HDI} use wild bootstrap for the symmetric settings. Also, apart from \texttt{SILM}, empirical performance is generally not affected by the scale of the covariates; however, in certain settings \texttt{SILM} performed very poorly when the covariates were not standardized. Thus, we standardize the covariates to benefit \texttt{SILM}.   

For each setting, we draw $\beta \in \mathbb{R}^{p}$ with $s = 4$ or $15$ \emph{active} (i.e., non-zero) coordinates drawn from the Rademacher distribution and set the remaining $p - s$ \emph{inactive} coordinates to 0. We arrange entries in $\beta$ in such that there is one active entry between two inactive entries (\emph{isolated}), one active between an active entry and an inactive entry (\emph{adjacent}), and one active entry between two other active entries (\emph{sandwiched}). We also use the same scheme for the inactive variables. We then set $Y = X \beta + \varepsilon$. 

To obtain $M_{\lambda^\star}$, we solve~\eqref{eq:selectM} to up to $500$ iterations using \texttt{fastclime}~\citep{fastClime} which starts with $\lambda = 1$ and uses warm starts to progressively shrink $\lambda$. We further select $\lambda^\star$ via \eqref{eq:minOverLambda} by using a grid search over the $\lambda$ values used by \texttt{fastclime} with $\delta = 10,000$. Empirically, a larger value of $\delta$ generally results in better coverage, but comes at the expense of confidence interval length; broadly speaking though, we see that for $\delta \geq 1,000$, the performance of the proposed procedure is fairly insensitive to the value of $\delta$.

Since in practice we do not know the appropriate Lasso tuning parameter $\lambda_1$ a priori, for the residual randomization procedure we employ the Square-Root Lasso~\citep{belloni2011sqrtLasso} implemented in \texttt{RPtests}~\citep{RPtests} to obtain estimates for $\hat{\beta}^l$. We follow \cite{zhang2017simultaneous} and rescale $\hat{\varepsilon}$ by $\sqrt{n / (n - \vert \hat{\beta}^l \vert_0)}$ as a finite-sample correction. For all settings, we use $1,000$ group actions/bootstrap resamples.

In Figures~\ref{fig:small} and \ref{fig:large} we show the empirical coverage and the confidence interval length for the active variables when $(n = 50, p = 100)$ for exchangeable and symmetric errors respectively. In the supplement, we provide corresponding figures for the inactive variables as well as results for $(n = 100, p = 300)$.

Among the competing methods, \texttt{SILM} performs the best, and has generally satisfactory performance across all settings except the Toeplitz case and the setting with Weibull covariates for $s = 4$. For $s = 4$, \texttt{HDI} generally outperforms \texttt{DLASSO}, but when $s = 15$, \texttt{HDI}'s performance decreases drastically. Generally, \texttt{BLPR} performs the worst and when $s = 15$ has coverage less than $.6$. All the competing methods perform poorly when the covariates have Toeplitz covariance. The sandwich coordinates typically have the lowest coverage and the isolated coordinates typically have empirical coverage closest to the nominal rate. 
 
In contrast, we see that \texttt{RR} nearly obtains the nominal 95\% coverage regardless of exchangeability or sign symmetry, and across all experimental configurations. This remarkable stability can be explained by our selection procedure for $M$ and the general properties of randomization tests~(see Introduction and Remark 1). At the same time, {\tt RR} typically yields larger interval lengths. While the interval length from {\tt RR} is generally longer than the competing methods, this is especially true when the covariates have Toeplitz covariance.
We posit that this is because solving ~\eqref{eq:selectM} for poorly conditioned sample covariances can yield large $\vert a^\top M\vert_1$ and thus results in larger confidence intervals.

In the supplement, we show the empirical coverage for the inactive coordinates with $(n = 50, p = 100)$ as well as all coordinates when $(n = 100, p = 300)$. For $(n = 100, p = 300)$, the results are qualitatively similar to the results for $(n = 50, p = 100)$; \texttt{RR} almost always has the best empirical coverage followed by \texttt{SILM}, \texttt{HDI}, \texttt{DLASSO}, and then \texttt{BLPR}. However, all methods (including \texttt{RR}) generally perform less well; in particular, in contrast to the $(n = 50, p = 100)$ case, there are some settings where \texttt{RR} does not attain nominal coverage. 
For the inactive variables when $(n = 50, p = 100)$ and $(n = 100, p = 300)$, all procedures except \texttt{BLPR} typically achieve (or exceed) nominal coverage. 

Table~\ref{tab:time} gives the average computation time (in seconds) required for each of the procedures with $n = 100$, $p = 300$. We form a confidence interval for each of the 300 coordinates using $1,000$ group actions/bootstrap resamples. \rev{Unsurprisingly}, \texttt{RR} requires more computational effort than competing procedures that appeal to asymptotic limiting distributions (\texttt{BLPR} and \texttt{DLASSO}). However, the computational effort is comparable to the resampling-based methods in our experiments (\texttt{HDI} and \texttt{SILM}). 
\begin{table}[h] \small
     \centering
          \caption{\label{tab:time}Mean Computation Time in Seconds.}
     \begin{tabular}{llllll}
     \toprule
     Method & \texttt{BLPR} & \texttt{HDI} & \texttt{DLASSO} & \texttt{SILM} & \texttt{RR}\\
     \midrule
     Time (s) & 18.9 & 374.2 & 11.2 & 61.1 & 135.1 \\
     \bottomrule
     \end{tabular}

\end{table}

\section{Discussion}
The theoretical guarantees coupled with the excellent empirical performance suggest that residual randomization is an appealing alternative for robust inference in high-dimensional linear models. Across a wide range of settings, it attains nominal coverage even when $n$ is small or the errors are heavy tailed. Of course, this does not come for free, as we observe that the confidence intervals produced are generally larger than those produced by competing methods---especially when 
there is strong cross-correlation in covariates. Nonetheless, we believe that in most practical applications \rev{slightly} larger confidence intervals is a small price to pay in exchange for better, more robust coverage.

While the procedure is asymptotically valid for any fixed $\delta \geq 1$, one practical concern is specifying a tuning parameter $\delta$ when selecting $M^\star$. If $\delta$ is too small, this may result in degraded empirical performance; however, the simulations show that there is a threshold of $\delta$ for which our procedure performs well across all settings. This threshold seems to coincide with a point at which the length of our confidence intervals become insensitive to increasing $\delta$. Thus, this threshold can be well approximated---at least via heuristics. 

Nonetheless, in Section 3 of the supplement, we describe an alternative procedure which also performs well empirically and is also asymptotically valid when an upper bound on $\Gamma$ in \eqref{eq:gamma} is known.

Several questions may be fruitful to pursue in the future. First, to form confidence intervals, we require control of $\vert \hat \beta^{l} - \beta \vert_1$. However, it would be interesting to investigate residual randomization procedures which only require small in-sample prediction error; i.e., $\frac{1}{n}\vert X(\hat\beta^{l} - \beta)\vert$. This would allow testing individual coordinates of $\beta$ without assuming stringent restricted eigenvalue or irrepresentability conditions.  There are additional advantages of the residual randomization framework which could be further exploited methodologically. For example, we primarily focused on selecting $M$, but one could also use the observed covariates $X$ to select specific invariances $G \in \mathcal{G}$ which optimize certain test properties. 

\section*{Acknowledgments}
This work was completed in part with resources provided by
the University of Chicago Research Computing Center.

\FloatBarrier
\clearpage

\bibliography{icml_bib}

\begin{thebibliography}{30}
\providecommand{\natexlab}[1]{#1}
\providecommand{\url}[1]{\texttt{#1}}
\expandafter\ifx\csname urlstyle\endcsname\relax
  \providecommand{\doi}[1]{doi: #1}\else
  \providecommand{\doi}{doi: \begingroup \urlstyle{rm}\Url}\fi

\bibitem[Belloni et~al.(2011)Belloni, Chernozhukov, and
  Wang]{belloni2011sqrtLasso}
Belloni, A., Chernozhukov, V., and Wang, L.
\newblock Square-root lasso: pivotal recovery of sparse signals via conic
  programming.
\newblock \emph{Biometrika}, 98\penalty0 (4):\penalty0 791--806, 2011.
\newblock ISSN 0006-3444.
\newblock \doi{10.1093/biomet/asr043}.
\newblock URL \url{https://doi.org/10.1093/biomet/asr043}.

\bibitem[Belloni et~al.(2016)Belloni, Chernozhukov, Hansen, and
  Kozbur]{belloni2016inference}
Belloni, A., Chernozhukov, V., Hansen, C., and Kozbur, D.
\newblock Inference in high-dimensional panel models with an application to gun
  control.
\newblock \emph{Journal of Business \& Economic Statistics}, 34\penalty0
  (4):\penalty0 590--605, 2016.

\bibitem[Bickel \& Freedman(1981)Bickel and Freedman]{bickel1981some}
Bickel, P.~J. and Freedman, D.~A.
\newblock Some asymptotic theory for the bootstrap.
\newblock \emph{The annals of statistics}, pp.\  1196--1217, 1981.

\bibitem[Bickel \& Freedman(1983)Bickel and Freedman]{bickel1983bootstrapping}
Bickel, P.~J. and Freedman, D.~A.
\newblock Bootstrapping regression models with many parameters.
\newblock \emph{Festschrift for Erich L. Lehmann}, pp.\  28--48, 1983.

\bibitem[B{\"u}hlmann \& Van De~Geer(2011)B{\"u}hlmann and Van
  De~Geer]{buhlmann2011statistics}
B{\"u}hlmann, P. and Van De~Geer, S.
\newblock \emph{Statistics for high-dimensional data: methods, theory and
  applications}.
\newblock Springer Science \& Business Media, 2011.

\bibitem[Cai et~al.(2011)Cai, Liu, and Luo]{cai2011clime}
Cai, T., Liu, W., and Luo, X.
\newblock A constrained {$\ell_1$} minimization approach to sparse precision
  matrix estimation.
\newblock \emph{J. Amer. Statist. Assoc.}, 106\penalty0 (494):\penalty0
  594--607, 2011.
\newblock ISSN 0162-1459.
\newblock \doi{10.1198/jasa.2011.tm10155}.
\newblock URL \url{https://doi.org/10.1198/jasa.2011.tm10155}.

\bibitem[Chatterjee \& Lahiri(2011)Chatterjee and
  Lahiri]{chatterjee2011bootstrap}
Chatterjee, A. and Lahiri, S.~N.
\newblock Bootstrapping lasso estimators.
\newblock \emph{J. Amer. Statist. Assoc.}, 106\penalty0 (494):\penalty0
  608--625, 2011.
\newblock ISSN 0162-1459.
\newblock \doi{10.1198/jasa.2011.tm10159}.
\newblock URL \url{https://doi.org/10.1198/jasa.2011.tm10159}.

\bibitem[Dezeure et~al.(2017)Dezeure, B\"{u}hlmann, and
  Zhang]{dezure2017simultaneous}
Dezeure, R., B\"{u}hlmann, P., and Zhang, C.-H.
\newblock High-dimensional simultaneous inference with the bootstrap.
\newblock \emph{TEST}, 26\penalty0 (4):\penalty0 685--719, 2017.
\newblock ISSN 1133-0686.
\newblock \doi{10.1007/s11749-017-0554-2}.
\newblock URL \url{https://doi.org/10.1007/s11749-017-0554-2}.

\bibitem[Fisher et~al.(1935)]{fisher1937design}
Fisher, R.~A. et~al.
\newblock The design of experiments.
\newblock \emph{The design of experiments.}, \penalty0 (2nd Ed), 1935.

\bibitem[Freedman \& Lane(1983{\natexlab{a}})Freedman and
  Lane]{freedman1983nonstochastic}
Freedman, D. and Lane, D.
\newblock A nonstochastic interpretation of reported significance levels.
\newblock \emph{Journal of Business \& Economic Statistics}, 1\penalty0
  (4):\penalty0 292--298, 1983{\natexlab{a}}.

\bibitem[Freedman \& Lane(1983{\natexlab{b}})Freedman and
  Lane]{freedman1983significance}
Freedman, D.~A. and Lane, D.
\newblock Significance testing in a nonstochastic setting.
\newblock In \emph{A {F}estschrift for {E}rich {L}. {L}ehmann}, Wadsworth
  Statist./Probab. Ser., pp.\  185--208. Wadsworth, Belmont, CA,
  1983{\natexlab{b}}.

\bibitem[Javanmard \& Montanari(2014)Javanmard and
  Montanari]{javanmard2014confidence}
Javanmard, A. and Montanari, A.
\newblock Confidence intervals and hypothesis testing for high-dimensional
  regression.
\newblock \emph{The Journal of Machine Learning Research}, 15\penalty0
  (1):\penalty0 2869--2909, 2014.

\bibitem[Kempthorne(1952)]{kempthorne1952design}
Kempthorne, O.
\newblock The design and analysis of experiments.
\newblock 1952.

\bibitem[Kuchibhotla \& Chakrabortty(2018)Kuchibhotla and
  Chakrabortty]{kuchibhotla2018moving}
Kuchibhotla, A.~K. and Chakrabortty, A.
\newblock Moving beyond sub-gaussianity in high-dimensional statistics:
  Applications in covariance estimation and linear regression.
\newblock \emph{arXiv preprint arXiv:1804.02605}, 2018.

\bibitem[Lehmann \& Romano(2006)Lehmann and Romano]{lehmann2006testing}
Lehmann, E.~L. and Romano, J.~P.
\newblock \emph{Testing statistical hypotheses}.
\newblock Springer Science \& Business Media, 2006.

\bibitem[Liu et~al.(2017)Liu, Xu, and Li]{liu2017bootstrap}
Liu, H., Xu, X., and Li, J.~J.
\newblock A bootstrap lasso+ partial ridge method to construct confidence
  intervals for parameters in high-dimensional sparse linear models.
\newblock \emph{arXiv preprint arXiv:1706.02150}, 2017.

\bibitem[Lopes(2014)]{lopes2014bootstrap}
Lopes, M.~E.
\newblock A residual bootstrap for high-dimensional regression with near
  low-rank designs.
\newblock In \emph{Proceedings of the 27th International Conference on Neural
  Information Processing Systems - Volume 2}, NIPS'14, pp.\  3239–3247,
  Cambridge, MA, USA, 2014. MIT Press.

\bibitem[Meinshausen et~al.(2009)Meinshausen, Meier, and
  B\"{u}hlmann]{meinshausen2009pvalues}
Meinshausen, N., Meier, L., and B\"{u}hlmann, P.
\newblock {$p$}-values for high-dimensional regression.
\newblock \emph{J. Amer. Statist. Assoc.}, 104\penalty0 (488):\penalty0
  1671--1681, 2009.
\newblock ISSN 0162-1459.
\newblock \doi{10.1198/jasa.2009.tm08647}.
\newblock URL \url{https://doi.org/10.1198/jasa.2009.tm08647}.

\bibitem[Pang et~al.(2014)Pang, Liu, and Vanderbei]{fastClime}
Pang, H., Liu, H., and Vanderbei, R.
\newblock The fastclime package for linear programming and large-scale
  precision matrix estimation in r.
\newblock \emph{Journal of Machine Learning Research}, 15\penalty0
  (14):\penalty0 489--493, 2014.
\newblock URL \url{http://jmlr.org/papers/v15/pang14a.html}.

\bibitem[Pitman(1937)]{pitman1937significance}
Pitman, E.~J.
\newblock Significance tests which may be applied to samples from any
  populations.
\newblock \emph{Supplement to the Journal of the Royal Statistical Society},
  4\penalty0 (1):\penalty0 119--130, 1937.

\bibitem[Rosenbaum(2003)]{rosenbaum2003exact}
Rosenbaum, P.~R.
\newblock Exact confidence intervals for nonconstant effects by inverting the
  signed rank test.
\newblock \emph{The American Statistician}, 57\penalty0 (2):\penalty0 132--138,
  2003.

\bibitem[Shah \& Buhlmann(2017)Shah and Buhlmann]{RPtests}
Shah, R. and Buhlmann, P.
\newblock \emph{RPtests: Goodness of Fit Tests for High-Dimensional Linear
  Regression Models}, 2017.
\newblock URL \url{https://CRAN.R-project.org/package=RPtests}.
\newblock R package version 0.1.4.

\bibitem[Tibshirani(1996)]{tibshirani1996Lasso}
Tibshirani, R.
\newblock Regression shrinkage and selection via the lasso.
\newblock \emph{J. Roy. Statist. Soc. Ser. B}, 58\penalty0 (1):\penalty0
  267--288, 1996.
\newblock ISSN 0035-9246.
\newblock URL
  \url{http://links.jstor.org/sici?sici=0035-9246(1996)58:1<267:RSASVT>2.0.CO;2-G&origin=MSN}.

\bibitem[Toulis(2019)]{toulis2019life}
Toulis, P.
\newblock Life after bootstrap: Residual randomization inference in regression
  models.
\newblock \emph{arXiv preprint arXiv:1908.04218}, 2019.

\bibitem[Van~de Geer et~al.(2014)Van~de Geer, B{\"u}hlmann, Ritov, Dezeure,
  et~al.]{vanDeGeer2014asymptotically}
Van~de Geer, S., B{\"u}hlmann, P., Ritov, Y., Dezeure, R., et~al.
\newblock On asymptotically optimal confidence regions and tests for
  high-dimensional models.
\newblock \emph{The Annals of Statistics}, 42\penalty0 (3):\penalty0
  1166--1202, 2014.

\bibitem[Vladimirova et~al.(2020)Vladimirova, Girard, Nguyen, and
  Arbel]{vladimirova2020sub}
Vladimirova, M., Girard, S., Nguyen, H., and Arbel, J.
\newblock Sub-weibull distributions: Generalizing sub-gaussian and
  sub-exponential properties to heavier tailed distributions.
\newblock \emph{Stat}, 9\penalty0 (1):\penalty0 e318, 2020.

\bibitem[Wasserman \& Roeder(2009)Wasserman and Roeder]{wasserman2009selection}
Wasserman, L. and Roeder, K.
\newblock High-dimensional variable selection.
\newblock \emph{Ann. Statist.}, 37\penalty0 (5A):\penalty0 2178--2201, 2009.
\newblock ISSN 0090-5364.
\newblock \doi{10.1214/08-AOS646}.
\newblock URL \url{https://doi.org/10.1214/08-AOS646}.

\bibitem[Zhang \& Zhang(2014)Zhang and Zhang]{zhang2014debiased}
Zhang, C.-H. and Zhang, S.~S.
\newblock Confidence intervals for low dimensional parameters in high
  dimensional linear models.
\newblock \emph{J. R. Stat. Soc. Ser. B. Stat. Methodol.}, 76\penalty0
  (1):\penalty0 217--242, 2014.
\newblock ISSN 1369-7412.
\newblock \doi{10.1111/rssb.12026}.
\newblock URL \url{https://doi.org/10.1111/rssb.12026}.

\bibitem[Zhang \& Cheng(2017)Zhang and Cheng]{zhang2017simultaneous}
Zhang, X. and Cheng, G.
\newblock Simultaneous inference for high-dimensional linear models.
\newblock \emph{J. Amer. Statist. Assoc.}, 112\penalty0 (518):\penalty0
  757--768, 2017.
\newblock ISSN 0162-1459.
\newblock \doi{10.1080/01621459.2016.1166114}.
\newblock URL \url{https://doi.org/10.1080/01621459.2016.1166114}.

\bibitem[Zou(2006)]{zou2006adaptive}
Zou, H.
\newblock The adaptive lasso and its oracle properties.
\newblock \emph{J. Amer. Statist. Assoc.}, 101\penalty0 (476):\penalty0
  1418--1429, 2006.
\newblock ISSN 0162-1459.
\newblock \doi{10.1198/016214506000000735}.
\newblock URL \url{https://doi.org/10.1198/016214506000000735}.

\end{thebibliography}
\bibliographystyle{icml2021}



%




\end{document}


\title{Supplement\\ \vspace{2mm} \normalsize{Robust Inference for High-Dimensional Linear Models via Residual Randomization}}
\date{}
\maketitle

In the supplement, we give proofs for statements in the main manuscript, propose an alternative procedure for selecting $M^\star$, and give additional simulation details.

\section{Proofs}
Recall the conditions required in the main text.

\begin{condition}[Covariates]\label{con:covariates}
Suppose that $X_{i,:} \in\mathbb{R}^p$ are generated i.i.d. with mean $0$ and covariance $\Sigma$. Let $\lambda_{\max}$ and $\lambda_{\min}$ denote the largest and smallest eigenvalues of $\Sigma$. Suppose each element of $X_{i,:}$ is sub-Weibull($\alpha$) and the de-correlated covariates $\tilde X_{i,:}^\top  = \Sigma^{-1/2}X_{i,:}^\top$ are jointly sub-Weibull($\alpha$) with \begin{equation}
\max\left(\Vert \tilde X_{i,:}^\top \Vert_{J,\Psi_\alpha}, \max_{v} \Vert X_{i,v} \Vert_{\Psi_\alpha} \right) \leq \kappa.
\end{equation}
Moreover, \begin{equation}
    \Gamma =  \max\left(\max_{v \in [p]}\E\left(\left[a^\top \Sigma^{-1}X_{i,:}^\top X_{i,v}\right]^2\right), \max_{v \in [p]}\E\left(\left[a^\top \Sigma^{-1}X_{i,:}^\top X_{j,v}\right]^2\right), \max_{u, v \in [p]^2} \E\left(\left[X_{i,u} X_{i,v}\right]^2\right)\right).
\end{equation}
\end{condition}

\begin{condition}[Sample Size]\label{con:highD}
Suppose $\kappa^\star =\kappa^2\max\left(\vert a \vert_2 \frac{ \sqrt{\lambda_{\max}}}{\sqrt{\lambda_{\min}}}, 1 \right)$, and
\begin{equation}
    n > \max\left\{\frac{4C_\alpha^2 (\kappa^\star)^2 \left[\log(2n)\right]^{4/\alpha}\left[3\log(pn)\right]^{4/\alpha - 1}}{\Gamma}, 64\Gamma(\log(pn) + 2 \log(p)) \right\}
\end{equation}
for some constant $C_\alpha$  which only depends on $\alpha$. 
\end{condition}

\begin{condition}[Exchangeability]\label{con:exchangeable}
Let $\mathcal{G} \subset \mathcal{G}_p$ where $\mathcal{G}_p$ is the set of all matrices corresponding to a permutation $g$ of $[n]$ such that (i) $[n] = N_1 \cup N_2$ for some $N_1$ and $N_2$ equal-sized disjoint sets,
and (ii) for all $j \in N_1$, $g(j) \in N_2$ and for all $j \in N_2$, $g(j) \in N_1$.
\end{condition}

\begin{condition}[Sign Symmetry]\label{con:symmetry}
Let $\mathcal{G} \subset \mathcal{G}_s$ where $\mathcal{G}_s$ is the set of all diagonal matrices containing only $\pm1$ such that there is an equal number of positive and negative $1$'s. 
\end{condition}

\begin{condition}[Cluster Exchangeability]\label{con:ClustExchangeable}
Suppose there exist $n_c$ disjoint sets $L_k$ with $[n] = \bigcup_k^{n_c} L_k$ and $|L_k| = n / n_c = J$ such that $\{\varepsilon_i\}_{i \in L_k}$ are exchangeable, but may otherwise be dependent. 
That is, $\mathcal{G} \subset \mathcal{G}_c$, where $\mathcal{G}_c$ is the set of all block diagonal matrices where the $G_{L_k, L_k}$ block is a permutation matrix satisfying Condition~\ref{con:exchangeable}. 
\end{condition}

\begin{condition}[Lasso with sub-Weibull errors]\label{con:lassSubWei}
Suppose $\varepsilon_i$ is sub-Weibull($\alpha$) with $\Vert \varepsilon_i \Vert_{\Psi_\alpha} \leq \kappa$. Suppose that
\begin{equation}
    \lambda_{\min} \geq 54 \min_{1 \leq h \leq p} \left\{ \Xi_{n,h} + \frac{32k\Xi_{n,h} }{h} \right\}
\end{equation}
where \begin{equation}
\begin{aligned}
    \Xi_{n,h} &= 14 \sqrt{2}\sqrt{\frac{\Upsilon_{n,h} h \log(36np/h)}{n}}\\
    &\quad + \frac{C_\alpha \kappa^2 h(\log(2n))^{2/\alpha}(h \log(36np/h))^{2/\alpha}}{n}\\
    \Theta_h &= \left\{\theta \in \mathbb{R}^p : \vert \theta \vert_0 \leq h, \vert \theta \vert_2 \leq 1\right\}\\
    \Upsilon_{n,h} &= \sup_{\theta \in \Theta_h}{\rm var}\left[\left(X_{i,:}^\top \theta\right)^2 \right].
\end{aligned} 
\end{equation}
Furthermore, suppose that the Lasso penalty term $\lambda_1$ is set such that
\begin{equation}
    \lambda_1 = 14 \sqrt{2}\sigma \sqrt{\frac{\log(np)}{n}} + \frac{C_{\alpha/2} \kappa^2(\log(2n))^{2/\alpha}(2\log(np))^{2/\alpha}}{n},
\end{equation}
and in addition to Condition~\ref{con:highD}
\begin{equation}
    n >  \frac{C_{\alpha/2}^2 \kappa^4 (\log(pn))^{8/\alpha -1}}{\sigma^2},
\end{equation}
where $\sigma = \max_{v \in [p]} \text{\rm var}(X_{i,v} \varepsilon_v)$ and $C_{\alpha/2}$ is a constant only depending on $\alpha$.
\end{condition}

\subsection*{Lemma 1}
\begin{lemma}\label{lem:wasserstein}
For any $M \in \mathbb{R}^{p \times p}$, let $d_1\left(F_{t}(X, \varepsilon),  F_{\hat t}(X, \varepsilon)\right)$ denote the Wasserstein-1 distance between the oracle randomization distribution and attainable randomization distributions. Then,
\begin{equation}\footnotesize
\begin{aligned}\label{eq:w1}
    d_1&\left(F_{t}(X, \varepsilon),  F_{\hat t}(X, \varepsilon)\right) \leq\left \vert \hat \beta^{l} - \beta \right \vert_1 \times \\
    &\quad \left[ \left\vert\sqrt{n}a^\top (I - M S)  \right \vert_\infty  + \left \vert a^\top M \right \vert_1 \E_Q \left(\left\vert X^\top  G X  /\sqrt{n} \right\vert_\infty\right) \right].
\end{aligned}
\end{equation}
where $Q$ is the uniform distribution over $G$ in $\mathcal{G}$. 
\end{lemma}
\begin{proof}
The debiased Lasso $\hat \beta^{d,M}$ is defined as
\begin{align}
\hat \beta^{d,M} = \hat \beta^{l} + \frac{1}{n} M X^\top (Y - X \hat \beta^{l})
\end{align}
so that
\begin{equation}
\begin{aligned}
\hat \beta^{d,M} - \beta &= \hat \beta^{l} - \beta  + \frac{1}{n} M X^\top (Y - X \hat \beta^{l}) + \frac{1}{n} M X^\top (Y - X  \beta) - \frac{1}{n} M X^\top (Y - X \beta)\\
&= \hat \beta^{l} - \beta  + \frac{1}{n} M X^\top X (\beta - \hat \beta^{l}) + \frac{1}{n} M X^\top \varepsilon\\
&= \left(I - \frac{1}{n} M X^\top X\right) (\hat \beta^l -  \beta) + \frac{1}{n} M X^\top \varepsilon.
\end{aligned}
\end{equation}
So that for any $M$, under the null hypothesis that $a^\top \beta = a_0$, we have 
\begin{equation} \begin{aligned}\label{eq:null}
    T_n &= \sqrt{n} (a^\top \hat{\beta}^{d, M} - a_0) \\
    &= \sqrt{n} a^\top(\hat{\beta}^{d, M} - \beta ) \\
    &=  \sqrt{n} a^\top (I - MS) (\hat \beta^{l} - \beta) + \frac{1}{\sqrt{n}} a^\top M X^\top \varepsilon.
\end{aligned}\end{equation}
Thus, the oracle randomization distribution which has access to the realization of $\varepsilon$ would be
\begin{equation}
t(G\varepsilon) = \sqrt{n}a^\top\left[    \left(I - \frac{1}{n} M X^\top X\right) (\hat \beta^l - \beta) + \frac{1}{n} M X^\top G \varepsilon\right]
\end{equation}
where $G$ is drawn uniformly from $\mathcal{G}$. The attainable randomization distribution which we actually use is
\begin{equation}\begin{aligned}\label{eq:residRandomDist}
    t(G\hat{\varepsilon}) = \frac{1}{\sqrt{n}} a^\top MX^\top  G \hat{\varepsilon} &= \frac{1}{\sqrt{n}} a^\top MX^\top  G (\varepsilon + X(\beta - \hat \beta^{l}))\\
    &=  \frac{1}{\sqrt{n}} a^\top MX^\top  G\varepsilon + \frac{1}{\sqrt{n}} a^\top MX^\top  GX(\beta - \hat \beta^{l})).
\end{aligned}\end{equation}

For any $Q$ which is a joint distribution over $(G_1, G_2)$ where, marginally, $G_1$ and $G_2$, are uniform from $\mathcal{G}$, we have 
\begin{equation}
    d_1\left(F_{t}(X, \varepsilon),  F_{\hat t}(X, \varepsilon)\right) \leq \E_Q \left( \vert t(G_1 \varepsilon) - t(G_2 \hat{\varepsilon}) \vert \right)
\end{equation}
Setting $Q$ to the distribution where $G_1 = G_2$ are drawn uniformly from $\mathcal{G}$, and  using~\eqref{eq:null} and~\eqref{eq:residRandomDist}, we have:
 \begin{equation}\begin{aligned}
    d_1\left(F_{t}(X, \varepsilon),  F_{\hat t}(X, \varepsilon)\right)
    &\leq \E_Q \left(\left\vert \frac{ a^\top M X^\top  G\varepsilon}{\sqrt{n}} + \sqrt{n} a^\top(I - MS)(\hat \beta^{l} - \beta)  \right. \right.\\
    & \qquad \qquad \left. \left. - \frac{ a^\top M X^\top  G\varepsilon}{\sqrt{n}} - \frac{a^\top M X^\top  G X (\beta - \hat \beta^{l})}{\sqrt{n}} \right\vert_1\right)  \\
&\leq \E_Q \left(\left\vert \sqrt{n}a^\top (I - MS)(\hat \beta^{l} - \beta) + \frac{a^\top M X^\top  G X (\hat \beta^{l} - \beta) }{\sqrt{n}} \right\vert_1\right) \\
&= \E_Q \left(\left\vert \sqrt{n}a^\top (I - MS) + a^\top M X^\top  G X  /\sqrt{n} \right\vert_\infty\right) \left \vert \hat \beta^{l} - \beta \right \vert_1 \\
&= \left[\frac{1}{\vert \mathcal{G} \vert } \sum_{G \in \mathcal{G}} \left\vert \sqrt{n}a^\top (I - MS) + a^\top M X^\top  G X  /\sqrt{n} \right\vert_\infty \right] \left[ \vert \beta^l - \beta\vert_1 \right] \\
&\leq \left[ \vert \beta^l - \beta\vert_1 \right]  \times \sqrt{n} \left( \left\vert a^\top (I - MS) \right\vert_\infty + \frac{1}{\vert \mathcal{G} \vert }  \sum_{G \in \mathcal{G}} \left \vert  \frac{1}{n} a^\top MX^\top  G X  \right\vert_\infty \right) \\
 \end{aligned}\end{equation}
\end{proof}

\subsection*{Lemma 2}

\begin{lemma}\label{lem:permNuisance}
Under Conditions~\ref{con:covariates} and~\ref{con:highD} and either Condition~\ref{con:exchangeable},~ \ref{con:symmetry}, or \ref{con:ClustExchangeable}, we have
\begin{equation}
P\left( \frac{1}{\vert \mathcal{G}\vert }\sum_{G \in \mathcal{G}}\vert X^\top G X\vert_\infty \geq  8 \sqrt{ \frac{2\Gamma (\log(pn) + 2\log(p))}{n}}\right) \leq  6\vert \mathcal{G}\vert (np)^{-1}.
\end{equation}
\end{lemma}
\begin{proof}
We bound $\vert X^\top G X\vert_\infty$ for each $G \in \mathcal{G}$, and then the final result follows from a union bound.

\paragraph{Exchangeability}
For some fixed $G \in \mathcal{G}$, letting $g(i) = \{j \, : \, G_{ij} \neq 0\}$ and $\gamma_i = \text{vec}(X_{i,:}^\top X_{g(i), :})$ where the $\text{vec}$ operator vectorizes the $p \times p$ matrix so that $\gamma_i \in \mathbb{R}^{p^2}$. Note that $ \E(\gamma_i) =  \E\left(\text{vec}(X_{i,:}^\top X_{g(i), :})\right) = 0$ since $i \neq g(i)$. Furthermore, 
\begin{equation}\begin{aligned}
    \Vert X_{i,u} X_{g(i),v} \Vert_{\psi_{\alpha/2}} &\leq \Vert X_{i,u}\Vert_{\Psi_\alpha} \Vert  X_{g(i),v} \Vert_{\Psi_\alpha} \leq \kappa^2 \leq \kappa^\star.
\end{aligned}
\end{equation}
Thus, each element of $\gamma_i$ is sub-Weibull($\alpha/2$) with Orlicz-norm bounded by $\kappa^\star$. Now,
\begin{equation}\begin{aligned}
    \left \vert \frac{1}{n}X^\top G X\right\vert_\infty &= \left \vert \frac{1}{n}\sum_i X_{i,:}^\top X_{g(i), :} \right \vert_\infty= \left \vert \frac{1}{n}\sum_i \gamma_i \right\vert_\infty,
\end{aligned}\end{equation}
but the $\gamma_i$'s are not independent of each other because $X_{i,:}$ appears both in $\gamma_i$ and $\gamma_{g^{-1}(i)}$. However, by construction, each $X_{i,:}$ only appears in one term of $\{\gamma_i\}_{i \in N_1}$ and one term in $\{\gamma_i\}_{i \in N_2}$. Thus, we can decompose the entire sum with possibly dependent terms into two separate sums of independent terms.
\begin{equation}\begin{aligned}
\left \vert \frac{1}{n}\sum_i \gamma_i \right\vert_\infty &= \left \vert \frac{1}{n}\sum_{i \in N_1} \gamma_i + \frac{1}{n}\sum_{i \in N_2} \gamma_i \right\vert_\infty\\
& \leq \max \left(\left \vert \frac{2}{n}\sum_{i \in N_1} \gamma_i\right\vert_\infty,  \left \vert \frac{2}{n}\sum_{i \in N_2} \gamma_i \right\vert_\infty\right)\\
\end{aligned}\end{equation}

We apply \citet[Theorem 3.4]{kuchibhotla2018moving} to each term with $t>0$ so that 
\begin{equation}\begin{aligned}
P&\left( \left \vert \frac{1}{n}\sum_i \gamma_i \right \vert_\infty \geq 7 \sqrt{ \frac{2\Gamma (t + 2\log(p))}{n}} + \frac{2C_\alpha \kappa^\star \log(n)^{2/\alpha}(t + 2\log(p))^{2/\alpha}}{n} \right)\\
& \leq P\left( \max \left(\left \vert \frac{2}{n}\sum_{i \in N_1} \gamma_i\right\vert_\infty,  \left \vert \frac{2}{n}\sum_{i \in N_2} \gamma_i \right\vert_\infty\right)
\geq 7 \sqrt{ \frac{2\Gamma (t + 2\log(p))}{n}} + \frac{2C_\alpha \kappa^\star \log(n)^{2/\alpha}(t + 2\log(p))^{2/\alpha}}{n}  \right)\\
&\leq 2P\left( \left \vert \frac{2}{n}\sum_{i \in N_1} \gamma_i\right\vert_\infty \geq 7 \sqrt{ \frac{2\Gamma (t + 2\log(p))}{n}} + \frac{2C_\alpha \kappa^\star \log(n)^{2/\alpha}(t + 2\log(p))^{2/\alpha}}{n}  \right)\\
&\leq 6\exp(-t).   
\end{aligned} 
\end{equation}
Note that when applying the concentration inequality to each term, the sample size is $n/2$ rather than $n$. Letting $t = \log(pn)$ and using Condition~\ref{con:highD} implies that the 
\begin{equation}
    7\sqrt{ \frac{2\Gamma (\log(pn) + 2\log(p))}{n}} + \frac{2C_\alpha \kappa^\star \log(n)^{2/\alpha}(\log(pn) + 2\log(p))^{2/\alpha}}{n} \leq 8\sqrt{ \frac{2\Gamma (\log(pn) + 2\log(p))}{n}}
\end{equation}
and using a union bound over all $G \in \mathcal{G}$ completes the proof.

\paragraph{Cluster Exchangeability}
Because $\mathcal{G}_c \subset \mathcal{G}_p$, then the proof for exchangeability directly implies that the statement holds for cluster exchangeability as well.  

\paragraph{Symmetry}
We repeat the same arguments with a slight modification due to sign-flipping. For some fixed $G \in \mathcal{G}$, let $\gamma_i = \text{vec}(G_{ii} X_{i,:} X_{i,:}^\top)$ so that $\gamma_i \in \mathbb{R}^{p^2}$. Note that $ \E(\gamma_i) =  G_{ii}\E\left(X_{i,:} X_{i}^\top\right) \neq 0$, so we instead pair together each $i \in N_1 = \{i : G_{ii} = 1\}$ with some $j \in N_2 = \{i : G_{ii} = -1\}$. Specifically, assume that $N_1$ and $N_2$ are ordered and let $N_1(i)$ and $N_2(i)$ denote the $i$th element of $N_1$ and $N_2$ respectively. We then define
\begin{equation}
    \tilde \gamma_i = \frac{1}{2}\left(\gamma_{N_1(i)} - \gamma_{N_2(i)}\right)
\end{equation}
so that $\E\left(\tilde \gamma_i \right) = 0$. Each element of $\tilde \gamma_i$ is sub-Weibull($\alpha/2$) with Orlicz-norm bounded by $\kappa^\star$. Now,
\begin{equation}\begin{aligned}
    \left \vert X^\top G X /n \right\vert_\infty &= \left \vert \frac{1}{n}\sum_i g_{ii}X_{i,:} X_{i}^\top \right \vert_\infty\\
    & = \left \vert \frac{2}{n}\sum_{i \in [n/2]} \tilde \gamma_i \right\vert_\infty.
\end{aligned}\end{equation}
Now we again apply \citet[Theorem 3.4]{kuchibhotla2018moving} to each term with $t>0$ so that 
\begin{equation}\begin{aligned}
P&\left( \left \vert \frac{2}{n}\sum_{i \in [n/2]} \tilde \gamma_i \right \vert_\infty \geq 7 \sqrt{ \frac{2\Gamma (t + 2\log(p))}{n}} + \frac{2C_\alpha \kappa^\star \log(n)^{2/\alpha}(t + 2\log(p))^{2/\alpha}}{n} \right)\\
&\leq 3\exp(-t).   
\end{aligned} 
\end{equation}

Again, letting $t = \log(pn)$ and using Condition~\ref{con:highD} implies that the 
\begin{equation}
    7\sqrt{ \frac{2\Gamma_s (\log(pn) + 2\log(p))}{n}} + \frac{2C_\alpha \kappa^\star \log(n)^{2/\alpha}(\log(pn) + 2\log(p))^{2/\alpha}}{n} \leq 8\sqrt{ \frac{2\Gamma (\log(pn) + 2\log(p))}{n}}.
\end{equation}
Taking a union bound over all $G \in \mathcal{G}$ completes the proof. Note that for sign-flips, we actually get a tighter upper bound on the probability, but we use looser leading term of $6$ from the permutation setting for simplicity.
\end{proof}

\subsection*{Lemma 3}
Suppose we select $M^\star = M_\lambda^\star$ by solving
\begin{equation}
\begin{aligned}
\label{eq:getLambda}
\lambda^\star = \arg\min_{\lambda \in [0,1) } \delta \vert a^\top (I - M_\lambda S)  \vert_\infty + \frac{\left \vert a^\top M_\lambda \right\vert_1}{\vert \mathcal{G} \vert} \sum_{G} \left\vert \frac{X^\top G X}{n}  \right\vert_\infty,
\end{aligned}
\end{equation}
where
\begin{equation}
\begin{aligned}
\label{eq:selectM}
M_\lambda &= \arg\min_M  \; \vert a^\top M \vert_1\\
&\quad \text{s.t.} \,  \left\vert a^\top (I - MS)\right \vert_\infty \leq \lambda.
\end{aligned}
\end{equation}

\begin{lemma}\label{lem:sigmaFeasible}
Under the Conditions~\ref{con:covariates} and \ref{con:highD}, we have
\begin{equation}\label{eq:sigmaFeasible}
P\left(\vert a^\top(I-\Sigma^{-1}S \vert_\infty \geq  8 \sqrt{ \frac{\Gamma (\log(pn) + 2\log(p))}{n}}\right) \leq  3(np)^{-1}.
\end{equation}
Thus, with probability at least $1 - 3(np)^{-1}$ the feasible set of \eqref{eq:selectM} is non-empty with $\lambda = 8 \sqrt{ \frac{\Gamma (\log(pn) + 2\log(p))}{n}}$ and 
\begin{equation}
    \vert a^\top M_\lambda \vert_1 \leq  \vert a^\top \Sigma^{-1} \vert_1.
\end{equation}

\end{lemma}
\begin{proof}
We show that~\eqref{eq:sigmaFeasible} holds which then trivially implies that $\Sigma^{-1}$ is in the feasible set for $\lambda = 8 \sqrt{ \frac{\Gamma (\log(pn) + 2\log(p))}{n}}$ and that $\vert a^\top M_\lambda \vert_1 < \vert a^\top \Sigma^{-1} \vert_1$ by the optimality of $M_\lambda$. 

Let $\gamma_i = a^\top(I - \Sigma^{-1}X_{i,:}^\top X_{i,:})$ such that $\gamma_i \in \mathbb{R}^p$. Note that $ \E(\gamma_i) = a^\top \E(I - \Sigma^{-1}X_{i,:}^\top X_{i,:}) = 0$. Furthermore, 
\begin{equation}\begin{aligned}
\Vert \gamma_{i,v} \Vert_{\psi_{\alpha/2}} &=
    \Vert (a^\top\Sigma^{-1}X_{i,:}^\top X_{i,:})_v \Vert_{\psi_{\alpha/2}}\\
    &= \Vert a^\top \Sigma^{-1/2} \tilde X_{i,:}^\top \left(\tilde X_{i,:} \Sigma^{1/2}\right)_v \Vert_{\psi_{\alpha/2}}\\
    &\leq \Vert a^\top \Sigma^{-1/2} \tilde X_{i,:}^\top \Vert_{\Psi_\alpha} \Vert   \left( \tilde X_{i,:} \Sigma^{1/2}\right)_v \Vert_{\Psi_\alpha} \leq \vert a \vert_2 \frac{ \sqrt{\lambda_{\max}}}{\sqrt{\lambda_{\min}}} \kappa^2 \leq \kappa^\star.
\end{aligned}
\end{equation}
Thus, each $\gamma_i$ is sub-Weibull($\alpha/2$) with Orlicz-norm bounded by $\kappa^\star$. Now, $a^\top(I-\Sigma^{-1}S) = \frac{1}{n}\sum_i \gamma_i$ so we again apply \citet[Theorem 3.4]{kuchibhotla2018moving} which implies that for any $t \geq 0$,  
\begin{equation}\label{eq:feasibleProb}
P\left( \left \vert \frac{1}{n}\sum_i \gamma_i \right \vert_\infty \geq 7 \sqrt{ \frac{\Gamma (t + 2\log(p))}{n}} + \frac{C_\alpha \kappa^\star \log(2n)^{2/\alpha}(t + 2\log(p))^{2/\alpha}}{n} \right) \leq  3\exp(-t).    
\end{equation}
Letting $t = \log(pn)$ and assuming Condition~\ref{con:highD}, we have
\begin{equation}\begin{aligned}
\frac{C_\alpha \kappa^\star \log(2n)^{2/\alpha}(\log(pn) + 2\log(p))^{2/\alpha}}{n} &= \sqrt{\frac{\Gamma(\log(pn) + 2\log(p))}{n}} \frac{C_\alpha \kappa^\star \log(2n)^{2/\alpha}(\log(pn) + 2\log(p))^{2/\alpha - 1/2}}{\sqrt{n \Gamma}}\\
&\leq \sqrt{\frac{\Gamma(\log(pn) + 2\log(p))}{n}}.
\end{aligned}
\end{equation}
Thus, the first term in the lower bound of~\eqref{eq:feasibleProb} dominates. Again, since $a^\top(I-\Sigma^{-1}\frac{1}{n}X^\top X) = \frac{1}{n}\sum_i \gamma_i$, we then have  
\begin{equation}
P\left( \left \vert a^\top\left(I-\Sigma^{-1}\frac{1}{n}X^\top X\right) \right \vert_\infty \geq 8 \sqrt{ \frac{\Gamma (\log(pn) + 2\log(p))}{n}} \right) \leq  3  (np)^{-1}.    
\end{equation}

\end{proof}

\begin{corollary}\label{cor:involvesBeta}
Assume the conditions of Lemma~\ref{lem:sigmaFeasible} and Lemma~\ref{lem:permNuisance}. Then with probability greater than $1 - 3(np)^{-1} -  6\vert \mathcal{G}\vert (np)^{-1}$ using $M^\star$ selected from \eqref{eq:selectM} and \eqref{eq:getLambda} yields
\begin{equation}\begin{aligned}
        d_1&\left(F_{t}(X, \varepsilon),  F_{\hat t}(X, \varepsilon)\right) \leq\left \vert \hat \beta^{l} - \beta \right \vert_1 \times \\
    &\quad\left[ 8\left(\delta + \left\vert a^\top \Sigma^{-1} \right \vert_1 \right) \sqrt{2\Gamma (\log(pn) + 2\log(p))} \right]. 
\end{aligned}\end{equation}
\end{corollary}

\begin{proof}
Let $b = 8 \sqrt{ \frac{\Gamma (\log(pn) + 2\log(p))}{n}}$ and suppose that $\Sigma^{-1}$ is in the feasible set for $\lambda = b$. Note, that by Condition~\ref{con:highD}, $b < 1 $. By the optimality of $\lambda^\star$, $M^\star$, and $M_b$ we have
\begin{equation}
    \begin{aligned}
         \delta \vert a^\top (I - M^\star S) \vert_\infty + \vert a^\top M^\star \vert_1 \E_Q \left(\left\vert X^\top  G X  /n \right\vert_\infty\right) &< \delta \vert a^\top (I - M_bS \vert_\infty + \vert a^\top M_b \vert_1 \E_Q \left(\left\vert X^\top  G X  /n \right\vert_\infty\right)\\
         &< \delta b + \vert a^\top \Sigma^{-1} \vert_1 \E_Q \left(\left\vert X^\top  G X  /n \right\vert_\infty\right).\\
    \end{aligned}
\end{equation}
Lemma~\ref{lem:permNuisance} and \ref{lem:sigmaFeasible} imply that with probability greater than $1 - 3(np)^{-1} -  6\vert \mathcal{G}\vert (np)^{-1}$ that $\Sigma^{-1}$ is feasible for $\lambda = b$ and that $\E_Q \left(\left\vert X^\top  G X  /n \right\vert_\infty\right) < \sqrt{2}b$. Applying Lemma~\ref{lem:wasserstein} then implies that
\begin{equation}\begin{aligned}
        d_1&\left(F_{t}(X, \varepsilon),  F_{\hat t}(X, \varepsilon)\right) \leq\left \vert \hat \beta^{l} - \beta \right \vert_1 \times \\
    &\quad\left[ 8\left(\delta + \left\vert a^\top \Sigma^{-1} \right \vert_1 \right) \sqrt{2\Gamma (\log(pn) + 2\log(p))} \right]. 
\end{aligned}\end{equation}
\end{proof}

\begin{theorem}[Sub-Weibull Errors and Covariates]\label{thm:lassSubWei}
Suppose Conditions~\ref{con:covariates}, \ref{con:highD}, and~\ref{con:lassSubWei} hold. Under either Condition~\ref{con:exchangeable} or \ref{con:symmetry}, with probability no less than $1 - \frac{6\vert \mathcal{G} \vert + 3}{np} + \frac{3}{np} + \frac{3}{n}$, 
\begin{equation}
\begin{aligned}
    d_1\left(F_{t}(X, \varepsilon),  F_{\hat t}(X, \varepsilon)\right) &\leq \frac{10752 s \sqrt{3\Gamma} \sigma}{\lambda_{\min}}(\delta + \vert a^T \Sigma^{-1}\vert_1)\frac{\log(np)}{\sqrt{n}} 
    \end{aligned}
\end{equation}
\end{theorem}

\begin{proof}
We combine Corollary~\ref{cor:involvesBeta} with results from \citet{kuchibhotla2018moving}. Specifically, to bound $\vert \beta^{l} - \beta \vert_1$, we apply \citet[Theorem 4.5]{kuchibhotla2018moving}, which states that with probability at least $1 - 3(np)^{-1} - 3n^{-1}$, letting the Lasso penalty parameter be:
\begin{equation}
    \lambda_1 = 14 \sqrt{2}\sigma \sqrt{\frac{\log(np)}{n}} + \frac{C_\gamma \kappa^2(\log(2n))^{2/\alpha}(2\log(np))^{2/\alpha}}{n}
\end{equation}
yields $\hat \beta^{l}$ such that
\begin{equation}
    \vert \hat \beta^{l} - \beta  \vert_2 \leq \frac{84 \sqrt{2}}{\lambda_{\min}} \left[\sigma \sqrt{\frac{s \log(np)}{n}} + \frac{C_{\alpha/2} \kappa^2 \sqrt{s}(\log(np))^{4/\alpha}}{n} \right].
\end{equation}
We require the corresponding bound on $\vert \hat \beta^{l} - \beta \vert_1$. As part of proving Theorem 4.5 (Appendix E.4), Kuchibolta and Chakroborty show that with the probability stated above,
\begin{equation}\label{eq:penaltyLargeEnough}
    \lambda_n \geq 2 \left\vert \frac{X^\top\varepsilon}{n} \right\vert_\infty.
\end{equation}
This allows us to apply \citet[Lemma 11.1]{hastie2015statistical} which states that when \eqref{eq:penaltyLargeEnough} holds, the estimation error belongs to the cone set:
\begin{equation}
\hat \nu = \hat \beta^{l} - \beta \in C(\mathcal{S},3) =  \left\{\nu \,:\, \vert \nu_{\mathcal{S}^C} \vert_1 \leq 3\vert \nu_{\mathcal{S}} \vert_1\right\},
\end{equation}
where $\mathcal{S} = \{j \, : \, \beta_j \neq 0\}$ and $\mathcal{S}^C$ is its complement. Thus, we have 
\begin{equation}\begin{aligned}
    \vert \hat \beta - \beta  \vert_1 &\leq \vert (\hat \beta - \beta )_{\mathcal{S}} \vert_1 + \vert (\hat \beta^{l})_{\mathcal{S}^C} \vert_1\\
    &\leq 4\vert (\hat \beta - \beta )_{\mathcal{S}} \vert_1 \leq 4\sqrt{s}\vert (\hat \beta - \beta )_{S} \vert_2 \leq 4\sqrt{s}\vert \hat \beta - \beta  \vert_2.
\end{aligned}\end{equation}
Thus, we have under the same conditions and stated probability that
\begin{equation}\begin{aligned}
\label{eq:lassoBound}
    \vert \hat \beta - \beta  \vert_1 &\leq \frac{336 s\sqrt{2}}{\lambda_{\min}} \left[\sigma  \sqrt{\frac{\log(np)}{n}} + \frac{C_{\alpha/2} \kappa^2 (\log(np))^{4/\alpha}}{n} \right] \\
    &\leq \frac{672 s\sqrt{2}}{\lambda_{\min}} \left[\sigma \sqrt{\frac{\log(np)}{n}}\right].
    \end{aligned}
\end{equation}
Combining with Corollary~\ref{cor:involvesBeta}, we then have with probability no less than $1 - \frac{6\vert \mathcal{G} \vert + 3}{np} - \frac{3}{np} - \frac{3}{n}$, 
\begin{equation}
\begin{aligned}
    d_1\left(F_{T}(X, \varepsilon),  F_{t}(X, \varepsilon)\right) &\leq \sqrt{n}\frac{672 s \sqrt{2}}{\lambda_{\min}} \left[\sigma \sqrt{\frac{\log(np)}{n}}\right]\times  8(\delta + \vert a^T \Sigma^{-1}\vert_1) \sqrt{\frac{6\Gamma \log(np)}{n}}
    \\
    &=\frac{10752 s \sqrt{3\Gamma} \sigma}{\lambda_{\min}}(\delta + \vert a^T \Sigma^{-1}\vert_1)\frac{\log(np)}{\sqrt{n}} 
    \end{aligned}
\end{equation}
\end{proof}

\clearpage
\section{Assumptions of \citet{belloni2016inference} for Cluster Dependence}

The proof of Theorem 2 follow directly from Corollary~\ref{cor:involvesBeta} and Theorem 1 of~\citet{belloni2016inference}. 

 Using our notation, we restate the relevant portion of Theorem 1 of \citet{belloni2016inference} as well as the conditions required. Recall that we assume that clusters are indexed by $i = 1, \ldots, n_c$ and observations within each cluster are indexed by $j = 1, \ldots, J$ so that $n = n_c J$. To accommodate this notation, we let $X_{ij} \in \mathbb{R}^p$ be the covariates of the $j$th observation from the $i$th cluster. Furthermore, let $X_{ijv} \in \mathbb{R}$ denote the $v$th covariate of the $j$th observation from the $i$th cluster. Similarly, let $Y_{ij}$ denote the $j$th outcome from the $i$th cluster.
 
\citet{belloni2016inference} begin with a more general additive fixed effects model where:
\begin{equation}
    Y_{ij} = f(w_{ij}) + e_i + \varepsilon_{ij} \qquad \text{ where } \qquad \E(\varepsilon_{ij} \mid w_{i1}, \ldots w_{iJ} ) = 0.
\end{equation}
However, the Approximately Sparse Model condition stated below requires that $f$ is well approximated by a linear model so that $f(w_{ij}) = X_{ij}^\top \beta + r(w_{ij})$ for some sparse $\beta$ and $r(w_{ij})$ term which vanishes as $p$ increases. We require the stronger assumption of a linear model; i.e., $r(w_{ij}) = 0$.

\citet{belloni2016inference} define the  ``demeaned observations''
 \begin{equation}\begin{aligned}
     \Ddot{X}_{ij} &= X_{ij} - \frac{1}{J} \sum_{j = 1}^J X_{ij}, \qquad
    \Ddot{Y}_{ij} &= Y_{ij} - \frac{1}{J} \sum_{j = 1}^J Y_{ij}, \qquad \text{ and } \qquad 
    \Ddot{\varepsilon}_{ij} &= \varepsilon_{ij} - \frac{1}{J} \sum_{j = 1}^J \varepsilon_{ij}.
 \end{aligned}\end{equation}
The Cluster-Lasso estimate is then defined as
 \begin{equation}\label{eq:belloniLasso}
     \hat \beta \in \arg\min_{b} \frac{1}{n_c J} \sum_{i = 1}^{n_c} \sum_{j = 1}^J (\Ddot{Y}_{ij} - \Ddot{X}_{ij}^\top b)^2 + \frac{\lambda_1}{n_c J} \sum_{v = 1}^p \hat \phi_v \vert b_v\vert,
 \end{equation}
where
\begin{equation}\label{eq:belloniPenalty}
    \lambda_1 = 2c \sqrt{n_c J} \Phi^{-1}(1 - \gamma / 2p)
\end{equation}
with $c >1$ being is a constant slack parameter, $\gamma = o(1)$, and $\Phi$ is the CDF of the standard Gaussian. Furthermore, $\hat \phi_v^2$ are estimates of
\begin{equation}
    \phi_v^2 = \frac{1}{n_c J} \sum_{i = 1}^{n_c} \left(\sum_{j =1}^J \Ddot{X}_{ijv} \Ddot{\varepsilon}_{ij} \right)^2.
\end{equation}
Since we do not have access to $\Ddot{\varepsilon}_{ij}$, we instead use 
\begin{equation}
    \hat \phi_v^2 = \frac{1}{n_c J} \sum_{i = 1}^{n_c} \left(\sum_{j =1}^J \Ddot{X}_{ijv} \hat{\varepsilon}_{ij} \right)^2
\end{equation}
where $\hat{\varepsilon}_{ij}$ are preliminary estimates of $\Ddot{\varepsilon}_{ij}$. \citet{belloni2016inference} give a procedure for calculating $\hat \phi_v$, but ultimately only require with probability $1 - o(1)$ for all $v \in [p]$ that
\begin{equation}\label{eq:belloniLoadings}
    l \phi_v \leq \hat \phi_v \leq u \phi_v
\end{equation}
for some $l \rightarrow 1$ and $u \leq C < \infty$.

The Sparse Eigenvalues condition concerns the empirical Gram matrix of the re-centered data
\begin{equation}
    \Ddot{M} = \{M_{jk}\}_{u,v \in [p]^2}, \qquad M_{u,v} = \frac{1}{n_c J}\sum_{i = 1}^{n_c} \sum_{j = 1}^J \Ddot{X}_{ijv} \Ddot{X}_{iju}, 
\end{equation}
and requires its minimum and maximum m-sparse eigenvalues to be bounded. Specifically, they require conditions on the quantities
\begin{equation}
\varphi_{\min}(m)(\Ddot{M}) = \min_{\delta \in \Delta(m)} \delta^\top  \Ddot{M} \delta \qquad \text{ and } \qquad \varphi_{\max}(m)(\Ddot{M}) = \max_{\delta \in \Delta(m)} \delta^\top  \Ddot{M} \delta
\end{equation}
where $\Delta(M) = \{\delta \in \mathbb{R}^p \, : \, \vert \delta \vert_0 \leq m, \vert \delta \vert_2 = 1\}$.

Finally, the regularity conditions require two additional quantities. The first, $\bar \omega_v$, involves the third moment of the $v$th covariate and error:
\begin{equation}
    \bar \omega_v = \left(\E \left[ \left \vert \frac{1}{\sqrt{J}} \sum_{j =1}^J \Ddot{X}_{ijv} \Ddot{\varepsilon}_{ij}  \right \vert^3  \right] \right)^{1/3}.
\end{equation}
They additionally require a measure of dependence within cluster, $\imath_J$:
\begin{equation}\footnotesize
    \imath_J = J \min_{1 \leq v \leq p} \frac{\E\left(\frac{1}{J} \sum_{j = 1}^J \ddot{X}^2_{ijv} \ddot{\varepsilon}^2_{ij}\right)}{\E\left(\frac{1}{J} \left[\sum_{j = 1}^J \ddot{X}^2_{ijv} \ddot{\varepsilon}^2_{ij}\right]^2\right)}.
\end{equation}
With no intra-cluster dependence, $\imath_J = J$, but in the worst case, $\imath_J = 1$.

\subsection*{Theorem 1 of \citet{belloni2016inference}}
Let $\{P_{n,J}\}$ be a sequence of probability laws, such that $\{(Y_{ij}, w_{ij}, X_{ij})\}_{j = 1}^J \sim P_{n,J}$, i.i.d. across i for which $n_c, J \rightarrow \infty$ jointly or $n_c \rightarrow \infty, J$ fixed. 
Suppose that Conditions ASM, SE, and R hold for probability measure $P = P_{P_{n,J}}$ induced by $P_{n,J}$. Consider a feasible Cluster-Lasso estimator with penalty level set by \eqref{eq:belloniPenalty} and penalty loadings obeying \eqref{eq:belloniLoadings}. Then
\begin{equation}
\vert \hat \beta - \beta \vert_1 = O_p\left(\sqrt{\frac{s^2 \log(p \vee n)}{n_c \imath_J}} \right).
\end{equation}

\paragraph{Condition ASM (Approximately Sparse Model)} The function $f(w_{ij})$ is well approximated by a linear combination of a dictionary of transformations, $X_{ij} = X_{n_c J}(w_{ij})$ where $X_{ij}$ is a $p \times 1$ vector with $p \gg n$ allowed, and $X_{n_c J}$ is a measureable map. That is, for each $i$ and $j$,
\begin{equation}
    f(w_{ij}) = X_{ij}^\top \beta + r(w_{ij}),
\end{equation}
where the coefficient $\beta$ and the remainder term $r(w_{ij})$ satisfy
\begin{equation}
    \vert \beta \vert_0 \leq s = o(n_c \imath_J ) \qquad \text{ and } \qquad \left[\frac{1}{n_c J} \sum_{i = 1}^{n_c} \sum_{j = 1}^J r(w_{ij})^2 \right]^{1/2} \leq A_s = O_p(\sqrt{s / n_c \imath_J}).
\end{equation}

\paragraph{Condition SE (Sparse Eigenvalues).} For any $C > 0$, there exists constants $0 < \kappa' < \kappa'' < \infty$, which do not depend on $n$ but may depend on $C$, such that with probability approaching one, as $n \rightarrow \infty$ $\kappa' \leq \varphi_{\min}(Cs)(\Ddot{M}) \leq \varphi_{\max}(Cs)(\Ddot{M}) \leq \kappa''$.

\paragraph{Condition R (Regularity Conditions).} Assume that for data $\{y_{ij}, w_{ij}\}$ that are i.i.d. across $i$, the following conditions hold with $X_{ij}$ defined as in Condition ASM with probability $1 - o(1)$:
\begin{enumerate}
    \item $\frac{1}{J} \sum_{j = 1}^J \E(\Ddot{X}_{ijv}^2 \Ddot{\varepsilon}_{ij}^2) + \left[\frac{1}{J} \sum_{j = 1}^J \E(\Ddot{X}_{ijv}^2 \Ddot{\varepsilon}_{ij}^2)\right]^{-1} = O(1)$
    \item $1 \leq \max_{v \in [p]} \phi_v / \min_{v \in [p]} \phi_v = O(1)$
    \item $1 \leq \max_{v \in [p]} \bar \omega_v / \sqrt{\E(\phi_v^2)} = O(1)$
    \item $\log^3(p) = o(n_cJ)$ and $s\log(p \vee n_cJ) = o(n_c \imath_J)$
    \item $\max_{v \in [p]} \vert \phi_v - \sqrt{\E(\phi_v^2)} \vert / \sqrt\E(\phi_v^2) = o(1)$.
\end{enumerate}

\clearpage

\section{Alternative Procedure for Selecting $M$}
Recall that
\begin{equation}\begin{aligned}\label{eq:clime}
M_{\lambda} &= \arg\min_M  \; \vert a^\top M \vert_1\\
&\quad \text{s.t.} \,  \left\vert a^\top (I - MS)\right \vert_\infty \leq \lambda.
\end{aligned}
\end{equation}

Define
\begin{equation}
\begin{aligned}
\label{eq:bounds}
d(\lambda) &=  \vert a^\top (I - M_\lambda S)  \vert_\infty + \frac{1}{\vert \mathcal{G} \vert} \sum_{G} \left\vert \frac{a^\top M_\lambda X^\top G X}{n}  \right\vert_\infty\\
d'(\lambda) &=  \vert a^\top (I - M_\lambda S) \vert_\infty + \vert a^\top M_\lambda \vert_1 \frac{1}{\vert \mathcal{G} \vert} \sum_{G} \left\vert \frac{X^\top G X}{n}  \right\vert_\infty\\
\end{aligned}
\end{equation}
such that $d(\lambda) \leq d'(\lambda)$. 
When $\Gamma$ (or some reasonable upper bound) is known, select $\delta_1$ so that $\delta_1 \geq 8 \sqrt{\Gamma}$ and $1 > \delta_1 \sqrt{ (\log(pn) + 2\log(p))/n}$. Condition~\ref{con:highD} ensures that such a $\delta_1$ exists. Then, an alternative way to select $M^\star$ is
\begin{equation}\begin{aligned}\label{eq:procI}
   \lambda^\star &= \min_{\lambda \in  [0, 1)}  \vert a^\top (I - M_\lambda S)\vert_\infty\\
    &\quad \text{s.t.} \, \text{\eqref{eq:clime} has non-empty feasible set for } \lambda \text{ and } d(\lambda) \leq d'(\delta_1 \sqrt{ (\log(pn) + 2\log(p))/n}). 
\end{aligned}
\end{equation}
Similar to the procedure described in the main text, \eqref{eq:procI} selects a $\lambda^\star$ which minimizes $\vert a^\top (I - M_\lambda S)  \vert_\infty$. 

This procedure, which we refer to as \texttt{RR Tuning Free}, may be preferable to the one (\texttt{RR}) presented in the main manuscript since it involves selecting a tuning parameter $\delta_1$ which is tied to a population quantity, $\Gamma$, rather than picking $\delta$ which may be hard to interpret. However, when $\delta_1$ is not large enough to satisfy $\delta_1 \geq 8 \sqrt{\Gamma}$, the procedure may not be asymptotically valid. This is in contrast to the original procedure which is asymptotically valid for any $\delta$, though the empirical performance may be affected by selecting $\delta$ too small.  

We show that this alternative selection procedure is also valid by slightly modifying the proof of Corollary~\ref{cor:involvesBeta}.

\begin{corollary}\label{cor:involvesBetaNew}
Assume the conditions of Lemma~\ref{lem:sigmaFeasible} and Lemma~\ref{lem:permNuisance}. Suppose in \eqref{eq:procI} that $\delta_1 \geq 8\sqrt{\Gamma}$ and $1>\delta_1 \sqrt{ (\log(pn) + 2\log(p))/n}$. Let $\delta_2 = \delta_1 /(8 \sqrt{\Gamma})$. Then with probability greater than $1 - 3(np)^{-1} -  6\vert \mathcal{G}\vert (np)^{-1}$ using $M^\star$ selected from \eqref{eq:procI} yields
\begin{equation}\begin{aligned}
        d_1&\left(F_{t}(X, \varepsilon),  F_{\hat t}(X, \varepsilon)\right) \leq\left \vert \hat \beta^{l} - \beta \right \vert_1 \times \\
    &\quad\left[ 8\left(\delta_2 + \left\vert a^\top \Sigma^{-1} \right \vert_1 \right) \sqrt{2\Gamma (\log(pn) + 2\log(p))} \right]. 
\end{aligned}\end{equation}
\end{corollary}

\begin{proof}
Let $b = \delta_1 \sqrt{ (\log(pn) + 2\log(p))/n}$. Suppose that $\Sigma^{-1}$ is in the feasible set for $\lambda = b$. By the optimality of $\lambda^\star$, $M^\star$, and $M_b$ we have
\begin{equation}
    \begin{aligned}
         \vert a^\top (I - M^\star S) \vert_\infty +  \E_Q \left(\left\vert a^\top M^\star X^\top  G X  /n \right\vert_\infty\right) &< \vert a^\top (I - M_b S \vert_\infty + \vert a^\top M_b \vert_1 \E_Q \left(\left\vert X^\top  G X  /n \right\vert_\infty\right)\\
         &< b + \vert a^\top \Sigma^{-1} \vert_1 \E_Q \left(\left\vert X^\top  G X  /n \right\vert_\infty\right).\\
    \end{aligned}
\end{equation}
Lemma~\ref{lem:permNuisance} and \ref{lem:sigmaFeasible} imply that with probability greater than $1 - 3(np)^{-1} -  6\vert \mathcal{G}\vert (np)^{-1}$ that $\Sigma^{-1}$ is feasible for $\lambda = b$ and that $\E_Q \left(\left\vert X^\top  G X  /n \right\vert_\infty\right) < \sqrt{2} b$. Applying Lemma~\ref{lem:wasserstein} then implies that
\begin{equation}\begin{aligned}
        d_1&\left(F_{t}(X, \varepsilon),  F_{\hat t}(X, \varepsilon)\right) \leq\left \vert \hat \beta^{l} - \beta \right \vert_1 \times \\
    &\quad\left[ 8\left(\delta_2 + \left\vert a^\top \Sigma^{-1} \right \vert_1 \right) \sqrt{2\Gamma (\log(pn) + 2\log(p))} \right]. 
\end{aligned}\end{equation}
\end{proof}

\clearpage

\section{Experiment Details}
We compare the empirical coverage of $95\%$ confidence intervals produced by \texttt{BLPR}~\citep{HDCI}, \texttt{HDI} \citep{hdi}, \texttt{SSLASSO}~\citep{javanmard2014confidence}\footnote{https://web.stanford.edu/~montanar/sslasso/code.html}, \texttt{SILM}~\citep{SILM} (non-studentized confidence intervals) and residual randomization (\texttt{RR}). For each setting, we replicate the experiment 1000 times for for $(n = 50, p = 100)$ and again for $(n = 100, p = 300)$. 

In each setting, we sample random $X \in \mathbb{R}^{n \times p}$ with rows drawn i.i.d. from either 
\begin{itemize}
    \item \textbf{N1}: $X_{i, :} \sim N(0, I)$
    \item \textbf{G1}: $X_{iv} \sim \text{Gamma}(1, 1) - 1$; i.e., a centered gamma with shape $ =1$ and rate $ =1$
    \item \textbf{N2}:  $X_{iv} \sim N(\mu ,1)$ with $P(\mu = -2) = P(\mu = 2) = 0.5$
    \item \textbf{NT}: $X_{i, :} \sim (0, \Sigma)$ for $\Sigma_{ij} = .8^{|i-j|}$
    \item \textbf{GT}: $X_{i, :} \sim  \text{Gamma}(\Sigma) - 1$ for $\Sigma_{ij} = .8^{|i-j|}$; i.e., each $X_{iv}$ is marginally a centered gamma with shape $ =1$ and rate $ =1$, but the covariance is Topelitz.
    \item \textbf{WB}: $X_{iv} \sim \text{Weibull}(1, 0.5) - \Gamma(2)$; i.e., a centered Weibull with scale $ =1$ and shape $ =1/2$.
\end{itemize} 

We sample the errors $\varepsilon \in \mathbb{R}^n$ from
\begin{itemize}
    \item \textbf{N1}: $\varepsilon_i \sim N(0, 1)$
    \item \textbf{N2}: $\varepsilon_i \sim N(\mu ,1)$ with $P(\mu = -2) = P(\mu = 2) = 0.5$;
    \item \textbf{HN}: $\varepsilon_i \sim N(0, 2 \Vert X_{i,:} \Vert_2^2 / p)$; i.e., the errors are \textbf{h}eteroskedastic and drawn from a \textbf{n}ormal distribution.
    \item \textbf{HM}: $\varepsilon_i \sim N(\mu, 2 \Vert X_{i,:} \Vert_2^2 / p)$ with $P(\mu = -2) = P(\mu = 2) = 0.5$; i.e., the errors are \textbf{h}eteroskedastic and drawn from a \textbf{m}ixture of normal distributions.
    \item \textbf{WB} $\varepsilon_i \sim \text{Weibull}(1, 0.5) - \Gamma(2)$; i.e., a centered Weibull with scale $ =1$ and shape $ =1/2$.  
\end{itemize}

For each setting, we draw $\beta \in \mathbb{R}^{p}$ with $s = 4$ or $15$ active (i.e., non-zero) coordinates drawn from the Rademacher distribution and set the remaining $p - s$ inactive coordinates to 0. We arrange entries in $\beta$ in such a way that there is one active entry between two inactive entries (isolated) so that $\beta_j = 1$ and $\beta_{j-1} = \beta_{j+1} = 0$, one active between an active entry and an inactive entry (adjacent) so that $\beta_j = \beta_{j-1} = 1$ and $\beta_{j+1} = 0$, and one active entry between two other active entries (sandwiched) $\beta_j = \beta_{j-1} = \beta_{j+1} = 1$. We also use the same scheme for the inactive variables. We then set $Y = X \beta + \varepsilon$.

Since in practice we do not know the appropriate tuning parameter $\lambda_1$ a priori, for the residual randomization procedure we employ the Square-Root Lasso~\citep{belloni2011sqrtLasso} implemented in \texttt{RPtests}~\citep{RPtests} to obtain estimates for $\hat{\beta}^l$. We follow \cite{zhang2017simultaneous} and rescale $\hat{\varepsilon}$ by $\sqrt{n / (n - \vert \hat{\beta}^l \vert_0)}$ as a finite-sample correction. 

Empirically, a larger value of $\delta$ generally results in \texttt{RR} producing better coverage at the expense of confidence interval length. We set $\delta = 10000$ for all settings; broadly speaking though, we see that for $\delta \geq 1000$, the performance of the proposed procedure is fairly insensitive to the value of $\delta$. 

Given Corollary~\ref{cor:involvesBetaNew} requires $\lambda_2 = \alpha \sqrt{(\log(pn) + \log(p))/n}$ for some $\alpha \geq 8 \sqrt{\Gamma}$. In practice, we may not know the value of $\Gamma$, but can provide a reasonable upper bound. Since we assumed that $8 \sqrt{\Gamma (\log(pn) + \log(p))/n} < 1$ and require that $\lambda^\star < 1$, in the implementation of \texttt{RR Tuning Free} used for the simulations, we set $\lambda = .99$. For added interpretability, we parameterize $\lambda$ with $\alpha \sqrt{\log(p))/n}$ and use \texttt{R}'s \texttt{optimize} function to find the smallest $\alpha \in [0.001, 0.99 / \sqrt{\log(p))/n}]$ whose $d(\lambda^\star) < d(0.99)$.   

Throughout our simulations, we use $1,000$ draws for the bootstrap-based methods, and $1,000$ group actions for our method. 

\clearpage

\section{Additional Experiments}\label{sec:additionalPlots}

\subsection{Inactive variables: $(n= 50, p = 100)$}
\label{sec:n_perm}
In Figures~\ref{fig:n_perm_small} and \ref{fig:n_sign_small}, we show empirical coverage and confidence interval length for the \emph{inactive variables} over 1000 trials when the errors and covariates are all sub-exponential with $(n = 50, p = 100)$ assuming exchangeable and sign symmetric errors. The same plots for the \emph{active variables} are shown in the main document.

Generally, all \texttt{DLASSO}, \texttt{SILM}, and \texttt{RR} achieve (or exceed) nominal coverage. \texttt{HDI} performs well under exchangeable errors, but generally undercovers in the symmetric setting. \texttt{BLPR} generally performs poorly in all settings.

\begin{figure}[h]
\centering
\centering
\includegraphics[scale=0.95]{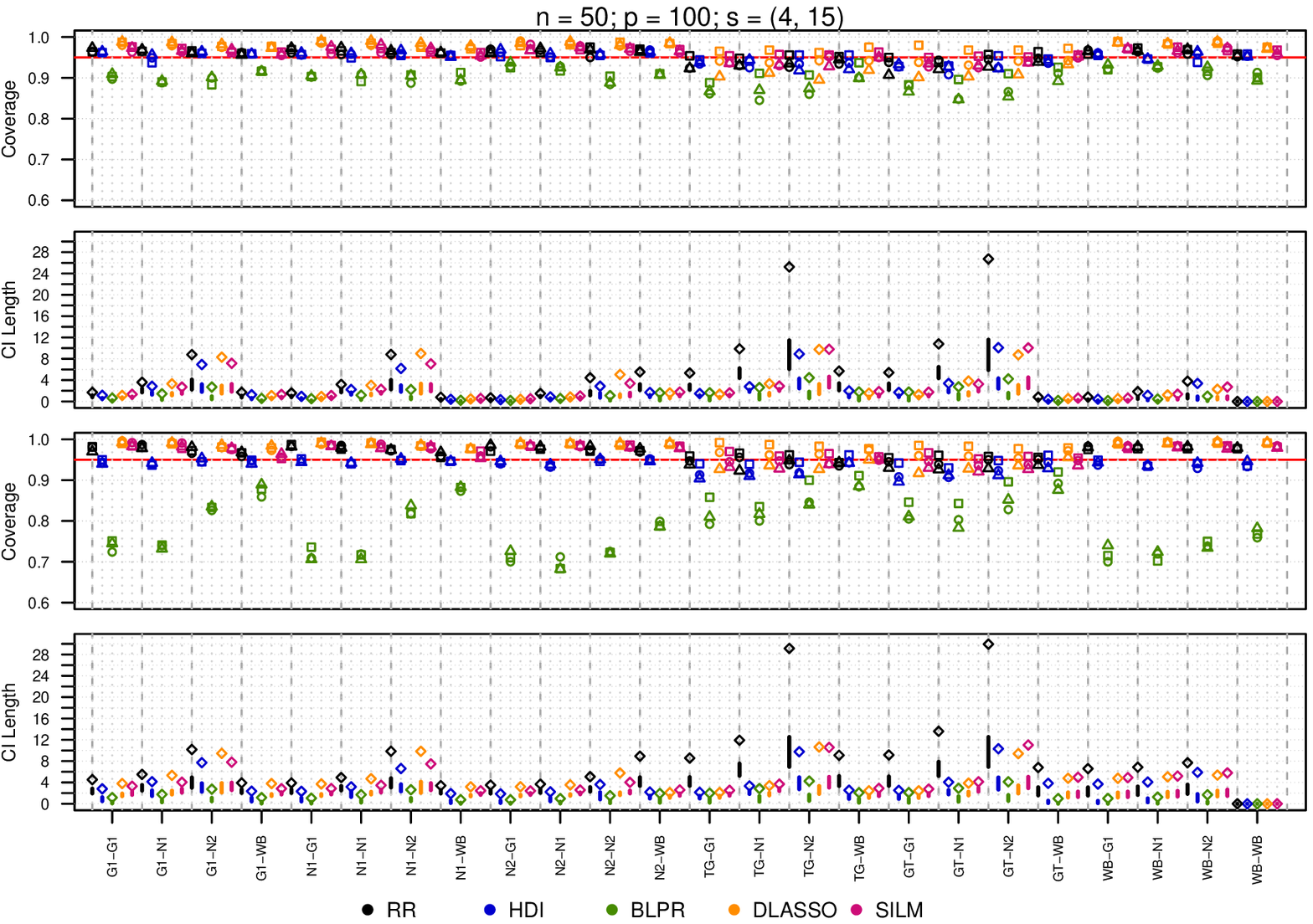}
\caption{Empirical coverage and confidence interval length for the \emph{inactive} variables with $n = 50$, $p = 100$, $1000$ replications and exchangeable errors. The top two panels are for $s=4$ and the bottom two are for $s = 15$. The first and third panels show empirical coverage rates for each procedure; the sandwich coordinate is denoted by $\Delta$, isolated is $\Box$, and adjacent is $\circ$. In the bottom panel, the line segment spans the $.25$ quantile and $.75$ quantile of the confidence interval lengths and the single point indicates the $.99$ quantile. Instead of showing the quantiles for each coordinate, we instead plot the maximum $.25$ (or $.75$, $.99$) quantile across the sandwich, isolated, and adjacent coordinates. The labels on the horizontal axis indicate a different simulation setting and are coded as ``Covariate - Errors'' where the different covariate and error settings are detailed in the main text. For some settings and procedures, the empirical coverage drops below $.6$ and is not shown.}
\label{fig:n_perm_small}
\end{figure}

\begin{figure}[t]
\centering
\includegraphics[scale=0.95]{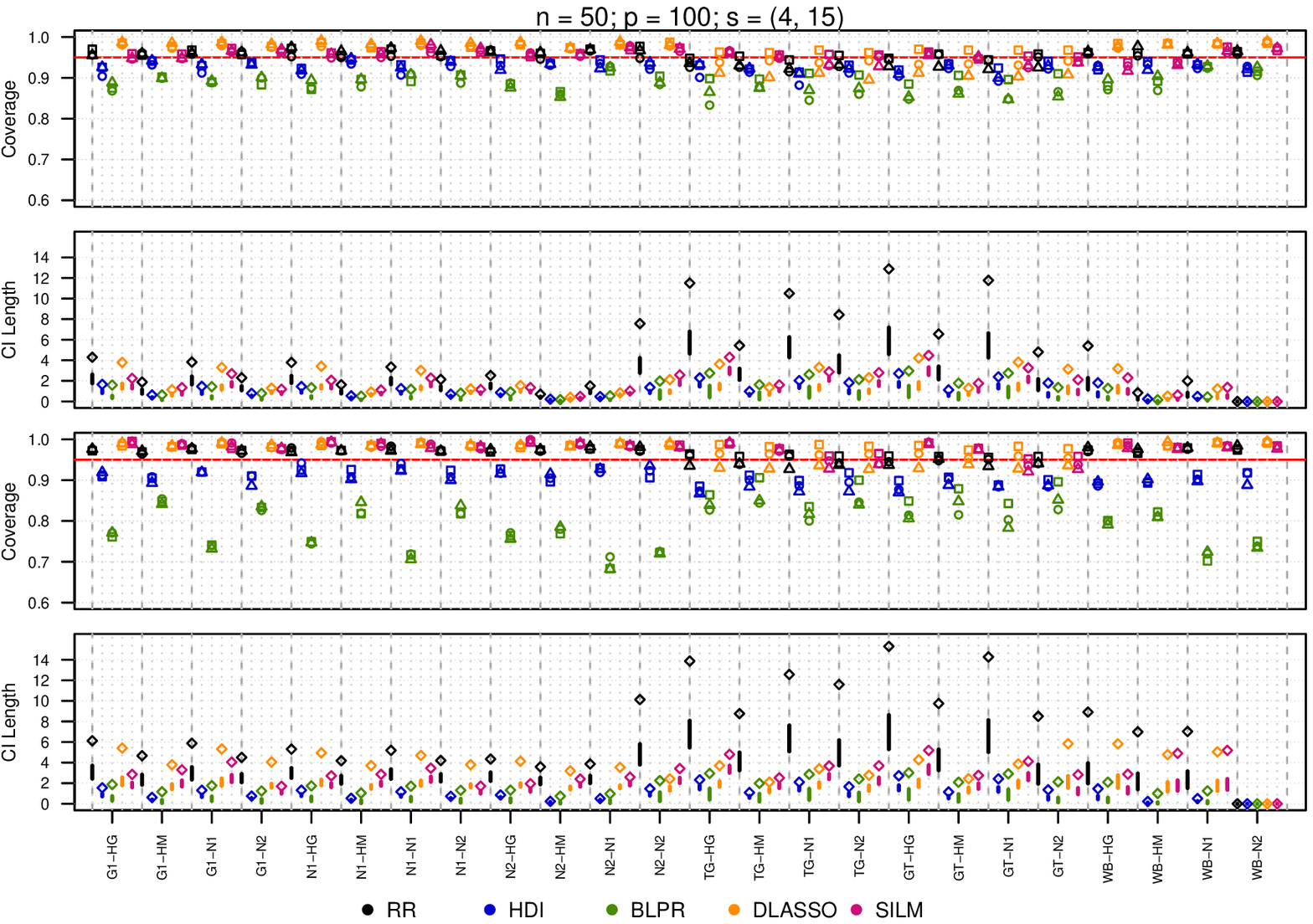}
\caption{Empirical coverage and confidence interval length for \emph{inactive variables} when $n = 50$ and $p = 100$ for \emph{sign symmetric errors}. All other elements remain the same as Figure~\ref{fig:n_perm_small}.}
\label{fig:n_sign_small}
\end{figure}


 \FloatBarrier
 \clearpage

\subsection{All variables: $(n= 100, p = 300)$}
In Figures~\ref{fig:a_perm_med} and \ref{fig:a_sign_med}, we show empirical coverage and confidence interval length for the active variables over 1000 trials when the errors and covariates are sub-exponential with $(n = 100, p = 300)$ assuming exchangeable and sign symmetric errors. Figures~\ref{fig:n_perm_med} and \ref{fig:n_sign_med} show the analogous plots for inactive variables.

The conclusions are qualitatively similar to the $(n = 50, p = 100)$ experiments for both active and inactive variables. However, we note that in this case the confidence intervals produced by the residual randomization procedure have lengths comparable to the competing methods in most settings.

\begin{figure}[htb]
\centering
\includegraphics[scale=0.95]{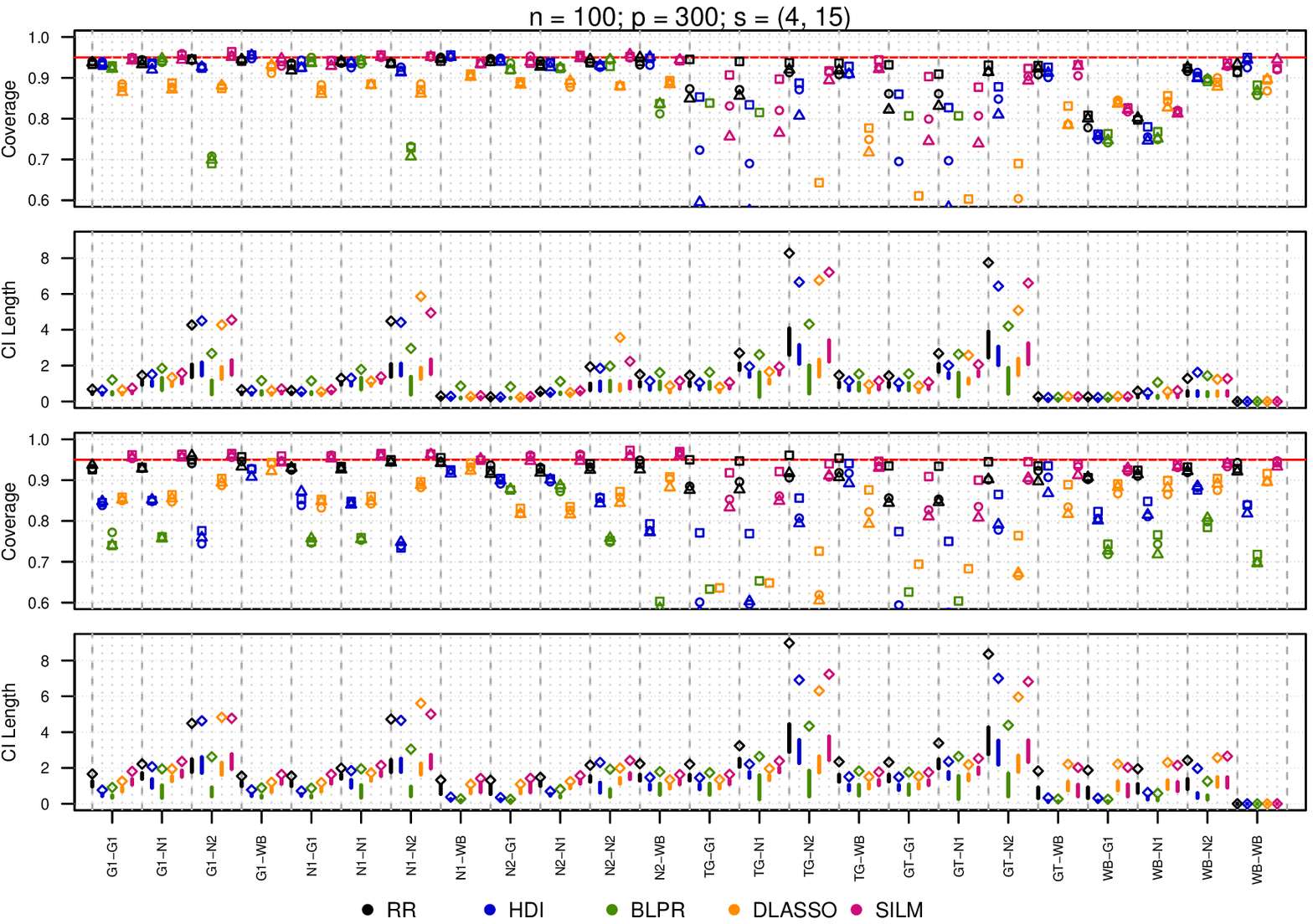}
\caption{Empirical coverage and confidence interval length for the active variables with $n = 100$, $p = 300$, $1000$ replications and exchangeable errors. The top two panels are for $s=4$ and the bottom two are for $s = 15$. The first and third panels show empirical coverage rates for each procedure; the sandwich coordinate is denoted by $\Delta$, isolated is $\Box$, and adjacent is $\circ$. In the bottom panel, the line segment spans the $.25$ quantile and $.75$ quantile of the confidence interval lengths and the single point indicates the $.99$ quantile. Instead of showing the quantiles for each coordinate, we instead plot the maximum $.25$ (or $.75$, $.99$) quantile across the sandwich, isolated, and adjacent coordinates. The labels on the horizontal axis indicate a different simulation setting and are coded as ``Covariate - Errors'' where the different covariate and error settings are detailed in the main text. For some settings and procedures, the empirical coverage drops below $.6$ and is not shown.}
\label{fig:a_perm_med}
\end{figure}

\begin{figure}[htb]
\centering
\includegraphics[scale=0.95]{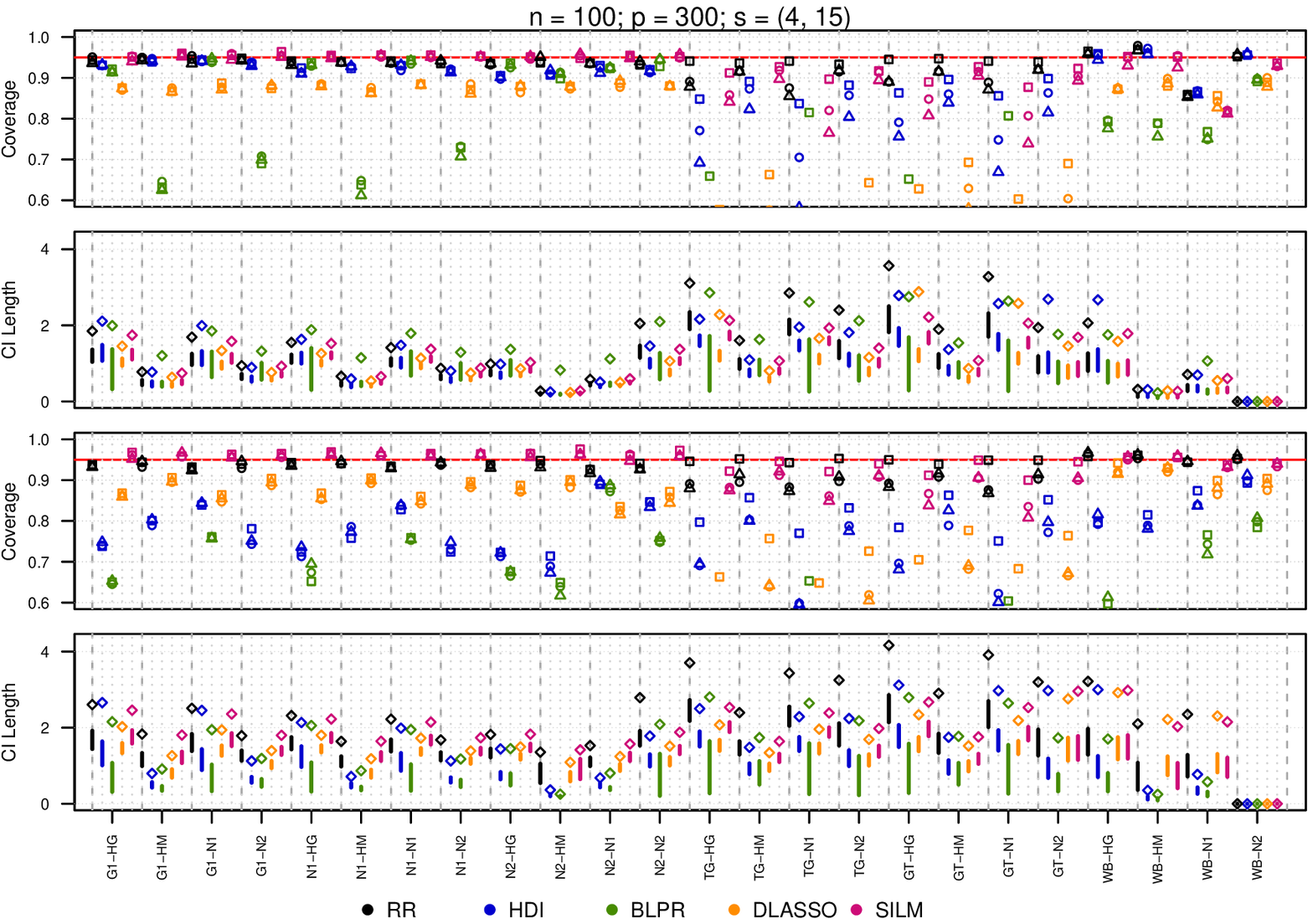}
\caption{Empirical coverage and confidence interval length for \emph{active variables} when $n = 100$ and $p = 300$ for \emph{sign symmetric errors}. All other elements remain the same as Figure~\ref{fig:a_perm_med}.}
\label{fig:a_sign_med}
\end{figure}

\begin{figure}[htb]
\centering
\includegraphics[scale=0.95]{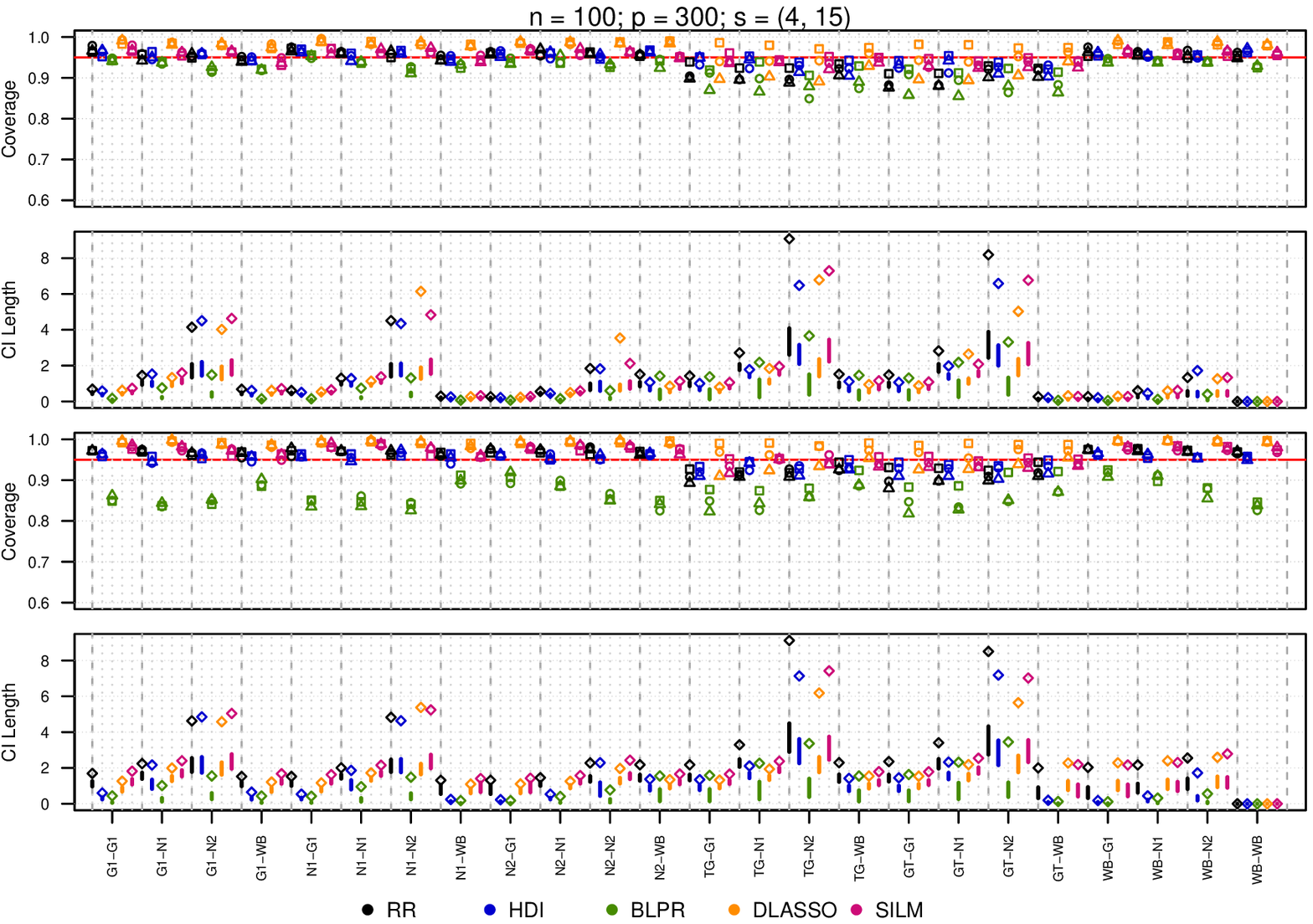}
\caption{Empirical coverage and confidence interval length for \emph{inactive variables} when $n = 100$ and $p = 300$ for \emph{exchangeable errors}. In the top two and bottom two panels, the true support of $\beta$ are $4$ and $15$ respectively. The first and third panels show empirical coverage rates for each procedure. In the bottom panels, the line segment indicates the $.25$ and $.75$ quantiles of the confidence interval lengths (averaged across all inactive variables for each run) and the single point indicates the $.99$ quantile. The labels on the horizontal axis indicate a different simulation setting and are coded as ``Covariate - Errors'' where the different covariate and error settings are detailed in the main text.}
\label{fig:n_perm_med}
\end{figure}
\begin{figure}\centering
\includegraphics[scale=0.95]{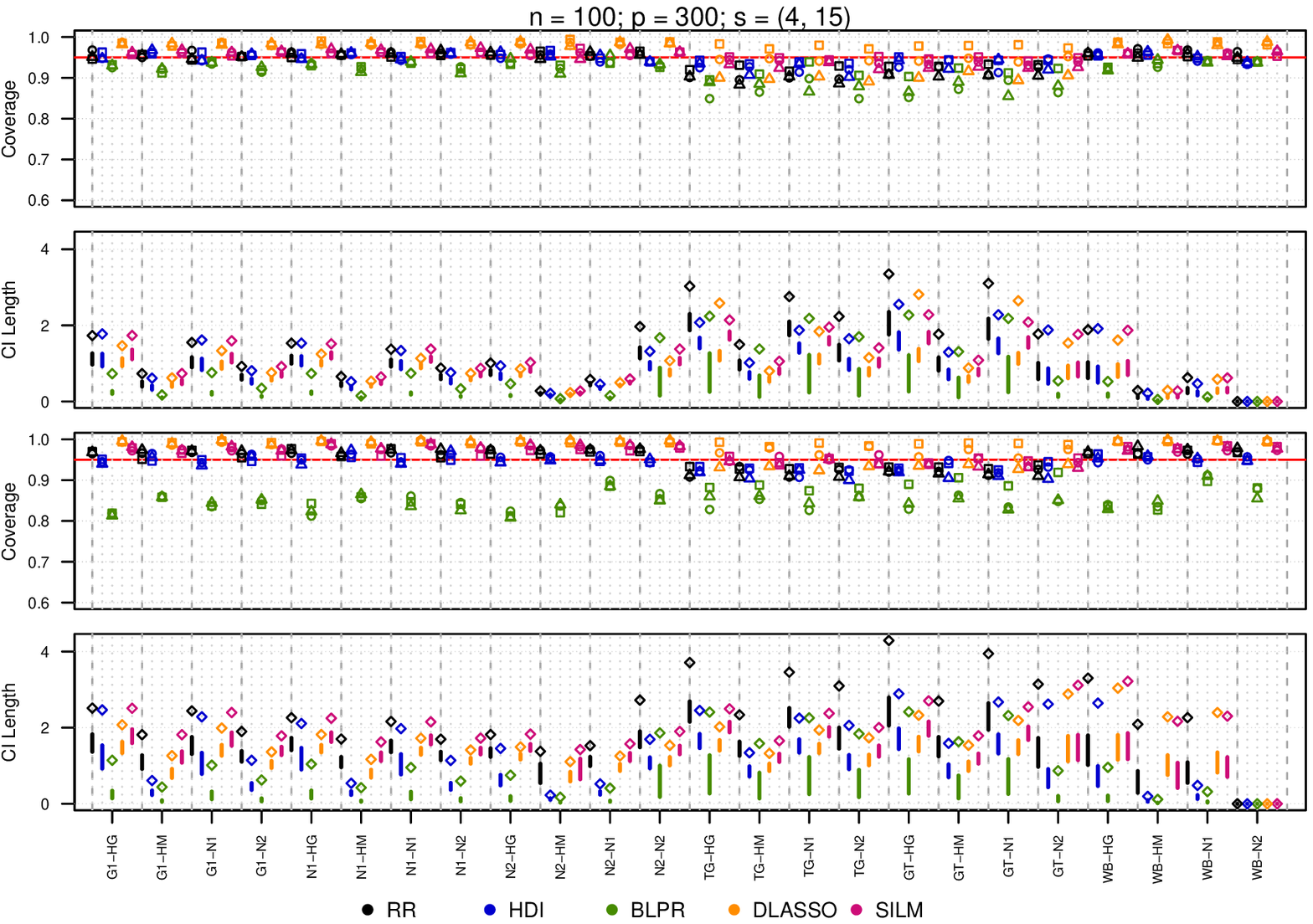}
\caption{Empirical coverage and confidence interval length for \emph{inactive variables} when $n = 100$ and $p = 300$ for \emph{sign symmetric errors}. All other elements remain the same as Figure~\ref{fig:n_perm_med}.}
\label{fig:n_sign_med}
\end{figure}

\clearpage

\subsection{Comparison of \texttt{RR} to \texttt{RR Tuning Free}}
In the following set of figures, we compare the performance of \texttt{RR} to \texttt{RR Tuning Free}. We note that when $s = 4$, \texttt{RR Tuning Free} generally performs slightly worse compared to \texttt{RR}. However, when $s = 15$ , \texttt{RR Tuning Free} performs comparably. In both cases, \texttt{RR Tuning Free} yields shorter CI lengths compared to \texttt{RR}, especially with covariates with Toeplitz covariances. With $\delta = 10000$, we would expect the solution from the selection procedure of the original method to be very close to what is obtained via that of \texttt{RR Tuning Free} albeit with a less precise grid search. The two main sources of discrepancies comes from 1) the second term in $d(\lambda)$ \eqref{eq:bounds} being a tighter upper bound compared to that in eq. 12 in the main text and 2) \texttt{fastclime} symmetrizing $M$.   

\begin{figure}[htb]
\centering
\includegraphics[scale=0.95]{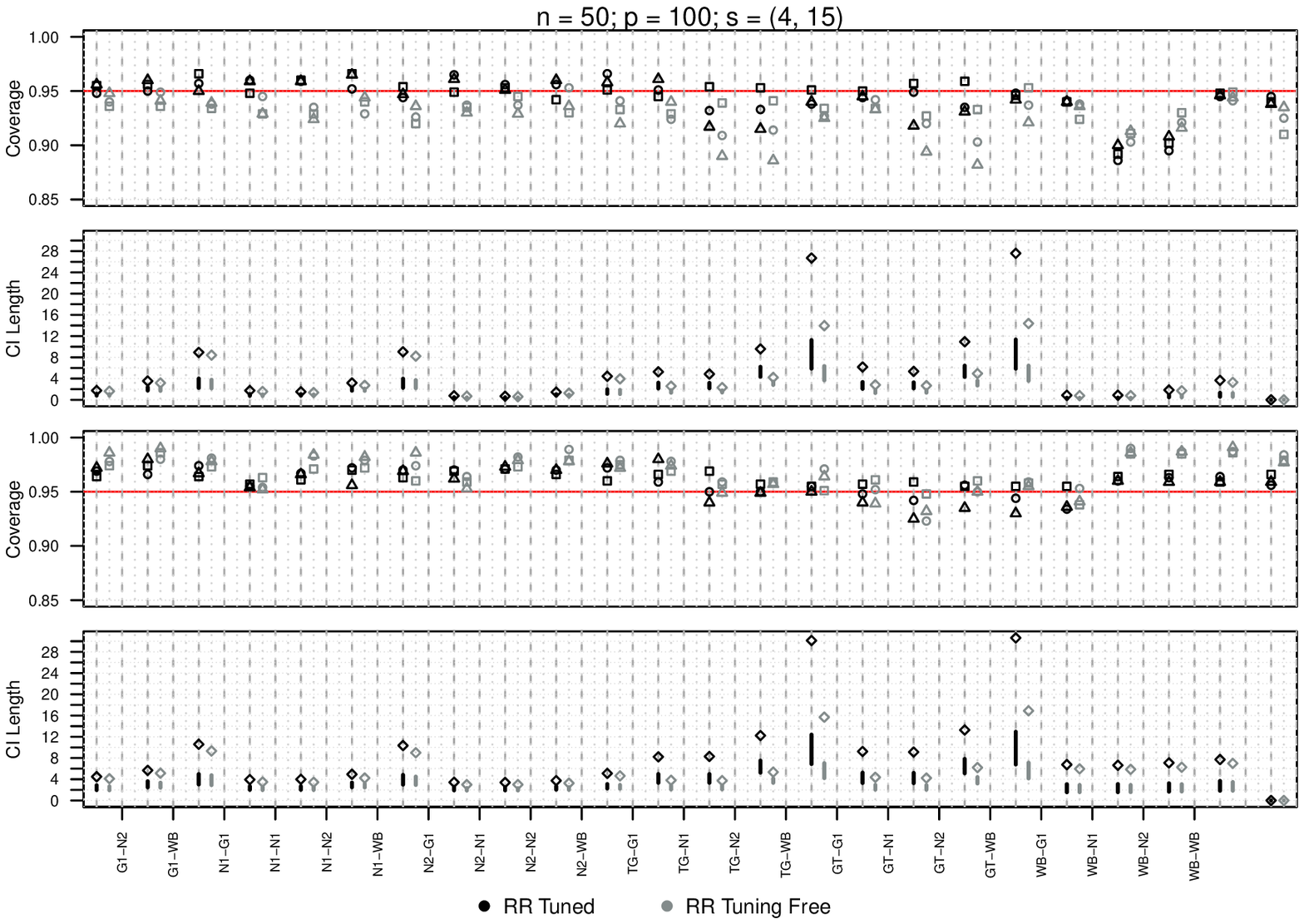}
\caption{Empirical coverage and confidence interval length for \emph{active variables} when $n = 50$ and $p = 100$ for \emph{exchangeable errors}. In the top two and bottom two panels, the true support of $\beta$ are $4$ and $15$ respectively. The first and third panels show empirical coverage rates for each procedure. In the bottom panels, the line segment indicates the $.25$ and $.75$ quantiles of the confidence interval lengths (averaged across all inactive variables for each run) and the single point indicates the $.99$ quantile. The labels on the horizontal axis indicate a different simulation setting and are coded as ``Covariate - Errors'' where the different covariate and error settings are detailed in the main text.}
\label{fig:rr_a_perm_small}
\end{figure}
\begin{figure}\centering
\includegraphics[scale=0.95]{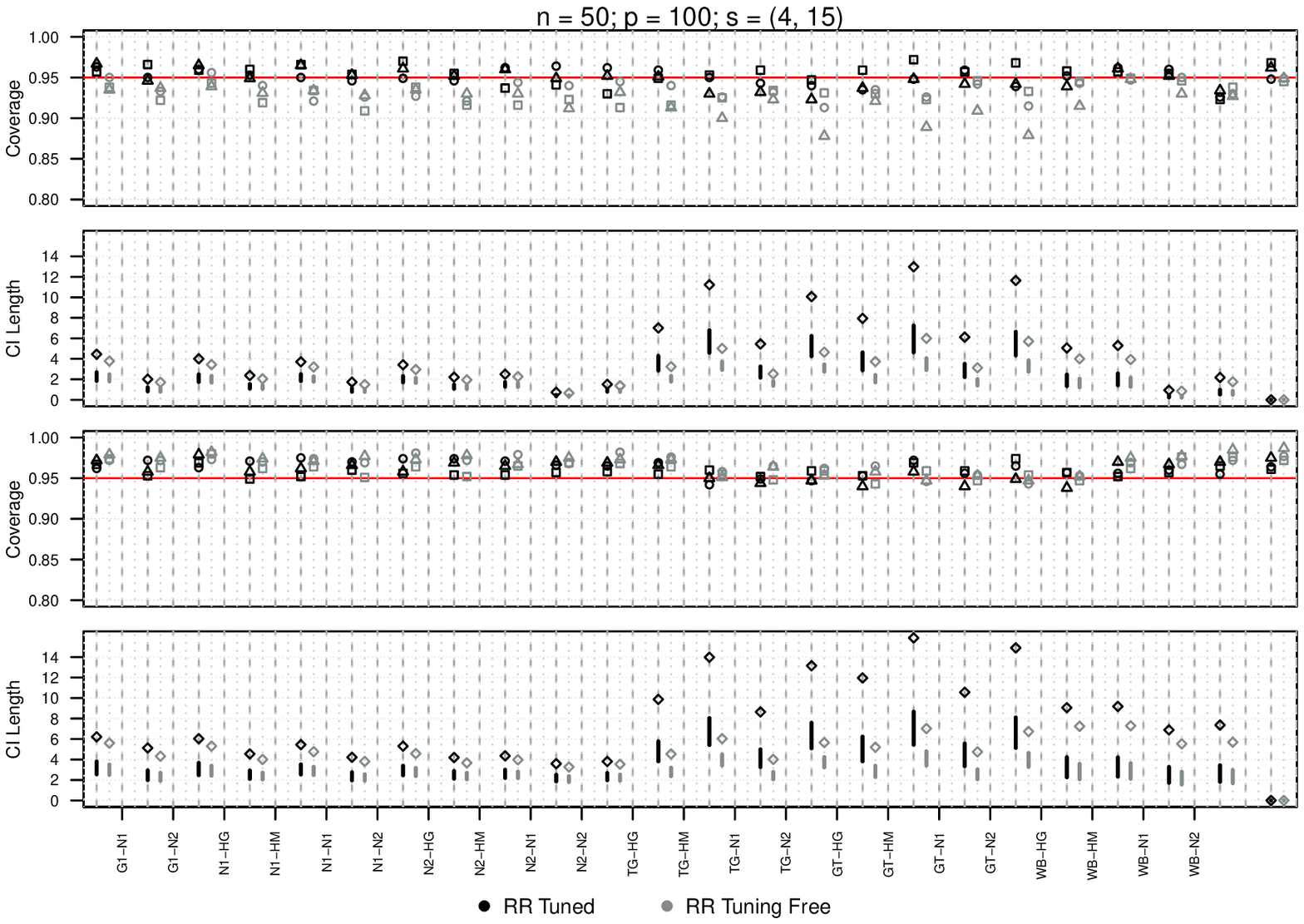}
\caption{Empirical coverage and confidence interval length for \emph{active variables} when $n = 50$ and $p = 100$ for \emph{sign symmetric errors}. All other elements remain the same as Figure~\ref{fig:rr_a_perm_small}.}
\label{fig:rr_a_sign_small}
\end{figure}

\begin{figure}[htb]
\centering
\includegraphics[scale=0.95]{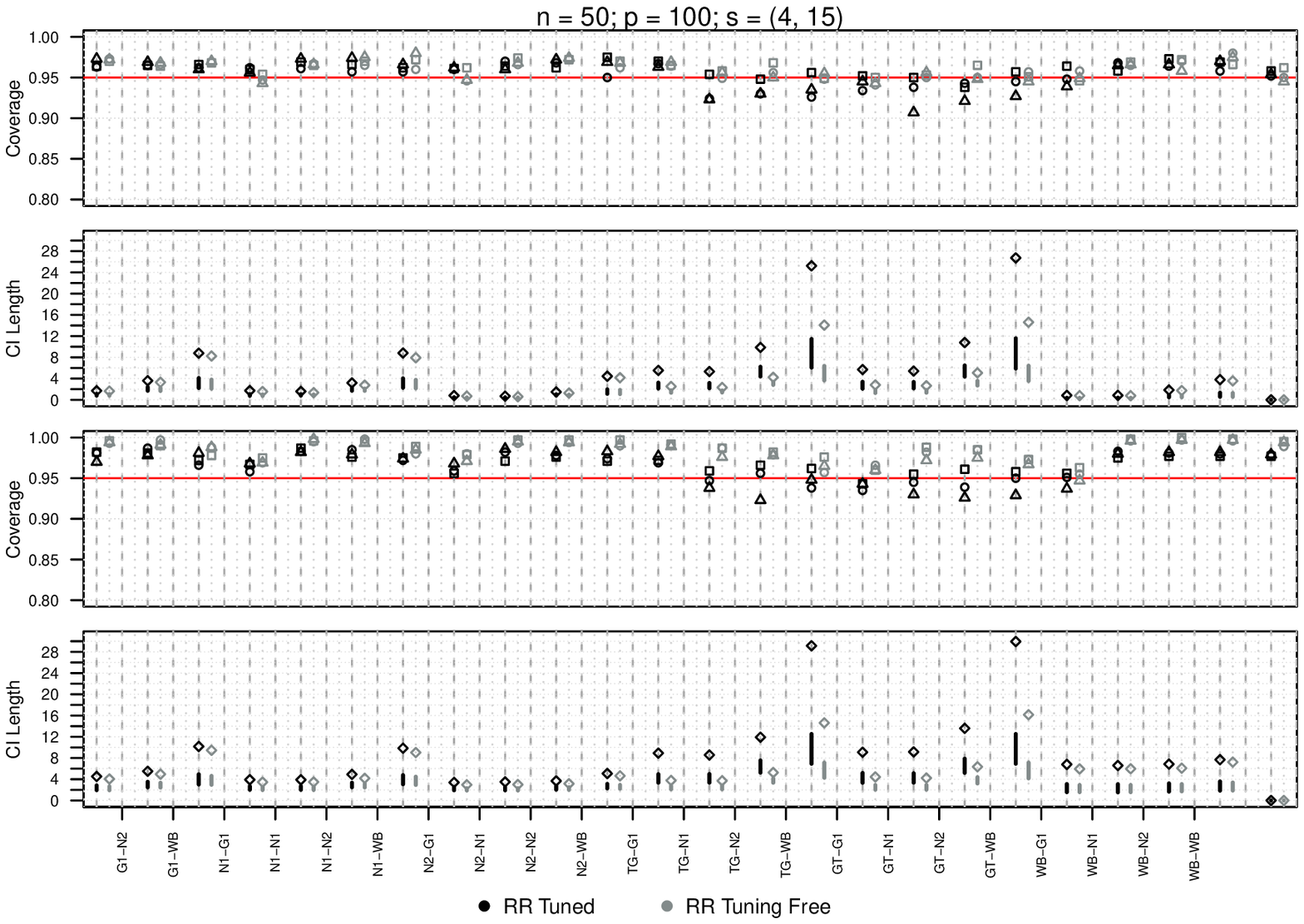}
\caption{Empirical coverage and confidence interval length for \emph{inactive variables} when $n = 50$ and $p = 100$ for \emph{exchangeable errors}. In the top two and bottom two panels, the true support of $\beta$ are $4$ and $15$ respectively. The first and third panels show empirical coverage rates for each procedure. In the bottom panels, the line segment indicates the $.25$ and $.75$ quantiles of the confidence interval lengths (averaged across all inactive variables for each run) and the single point indicates the $.99$ quantile. The labels on the horizontal axis indicate a different simulation setting and are coded as ``Covariate - Errors'' where the different covariate and error settings are detailed in the main text.}
\label{fig:rr_n_perm_small}
\end{figure}
\begin{figure}\centering
\includegraphics[scale=0.95]{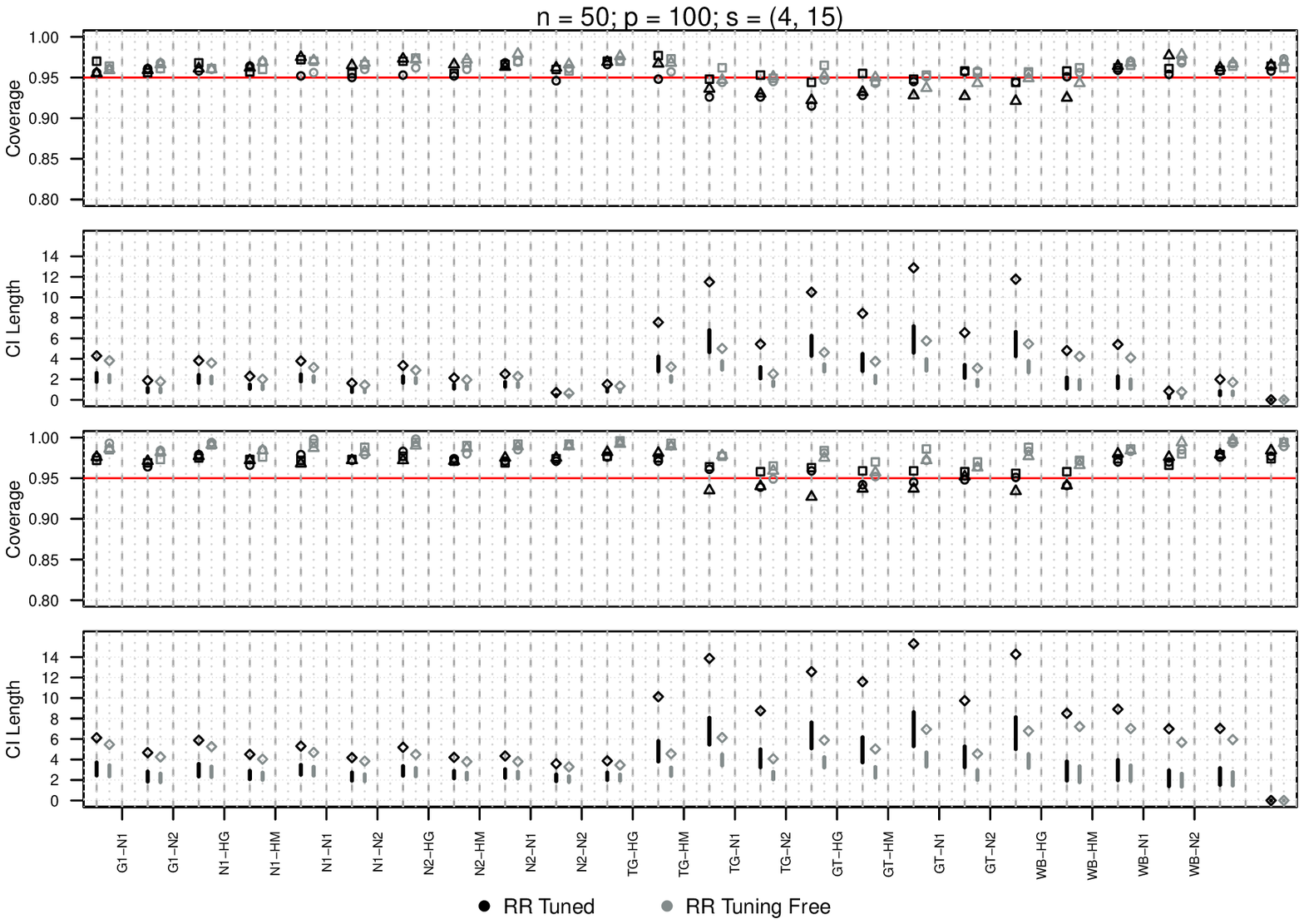}
\caption{Empirical coverage and confidence interval length for \emph{inactive variables} when $n = 50$ and $p = 100$ for \emph{sign symmetric errors}. All other elements remain the same as Figure~\ref{fig:rr_n_perm_small}.}
\label{fig:rr_n_sign_small}
\end{figure}

\begin{figure}[htb]
\centering
\includegraphics[scale=0.95]{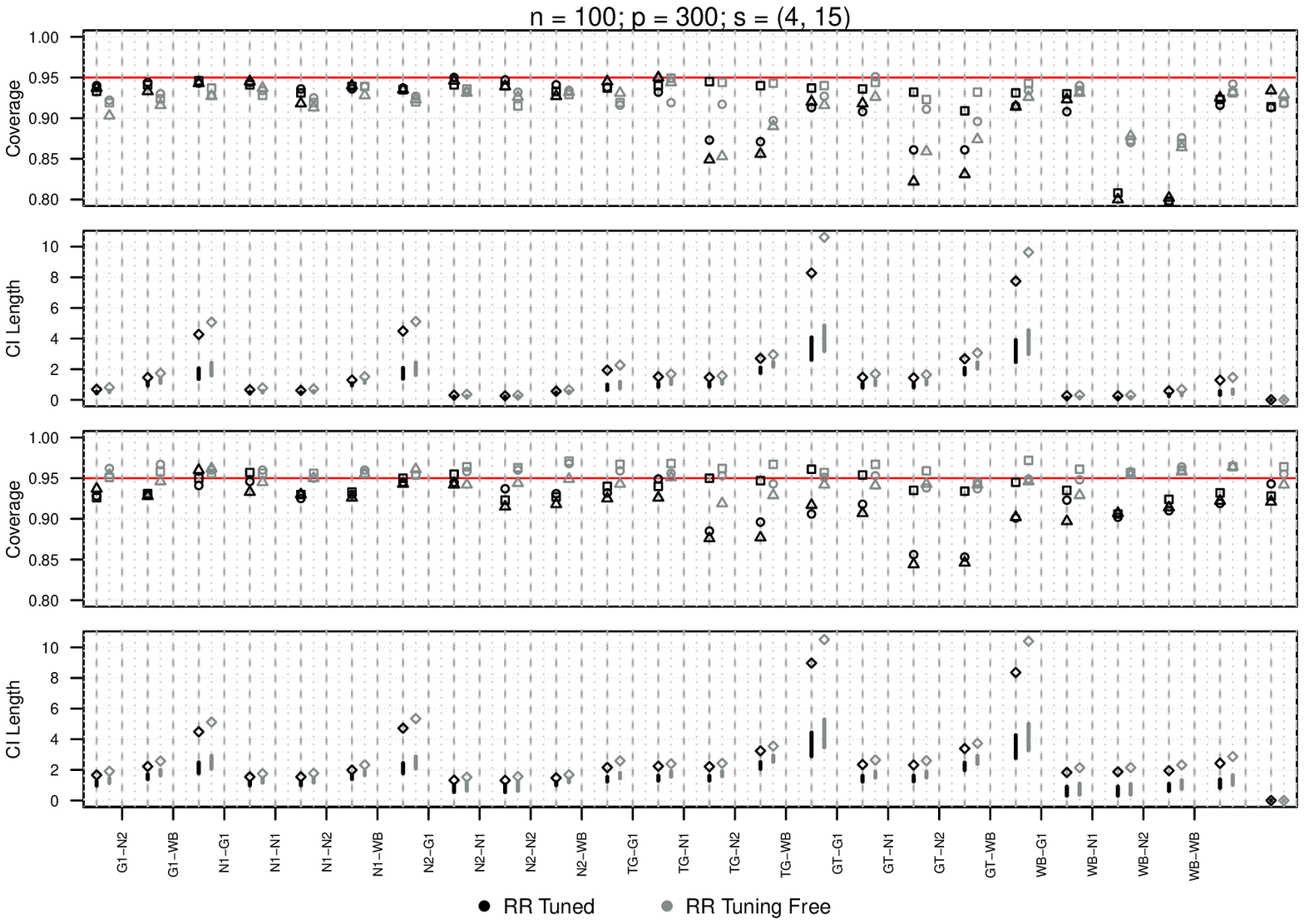}
\caption{Empirical coverage and confidence interval length for \emph{active variables} when $n = 100$ and $p = 300$ for \emph{exchangeable errors}. In the top two and bottom two panels, the true support of $\beta$ are $3$ and $10$ respectively. The first and third panels show empirical coverage rates for each procedure. In the bottom panels, the line segment indicates the $.25$ and $.75$ quantiles of the confidence interval lengths (averaged across all inactive variables for each run) and the single point indicates the $.99$ quantile. The labels on the horizontal axis indicate a different simulation setting and are coded as ``Covariate - Errors'' where the different covariate and error settings are detailed in the main text.}
\label{fig:rr_a_perm_med}
\end{figure}
\begin{figure}\centering
\includegraphics[scale=0.95]{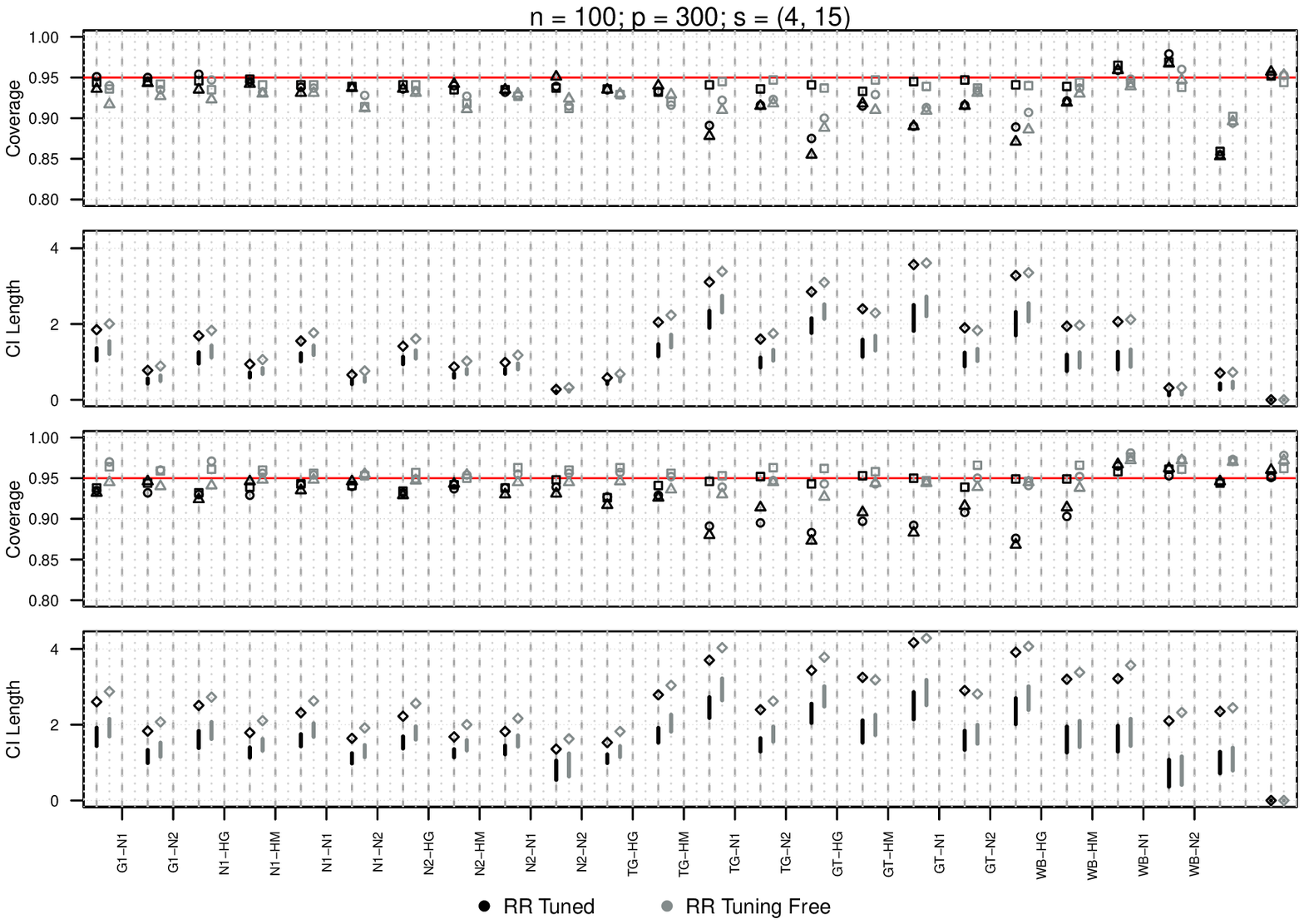}
\caption{Empirical coverage and confidence interval length for \emph{active variables} when $n = 100$ and $p = 300$ for \emph{sign symmetric errors errors}. All other elements remain the same as Figure~\ref{fig:rr_a_perm_med}.}
\label{fig:rr_a_sign_med}
\end{figure}

\begin{figure}[htb]
\centering
\includegraphics[scale=0.95]{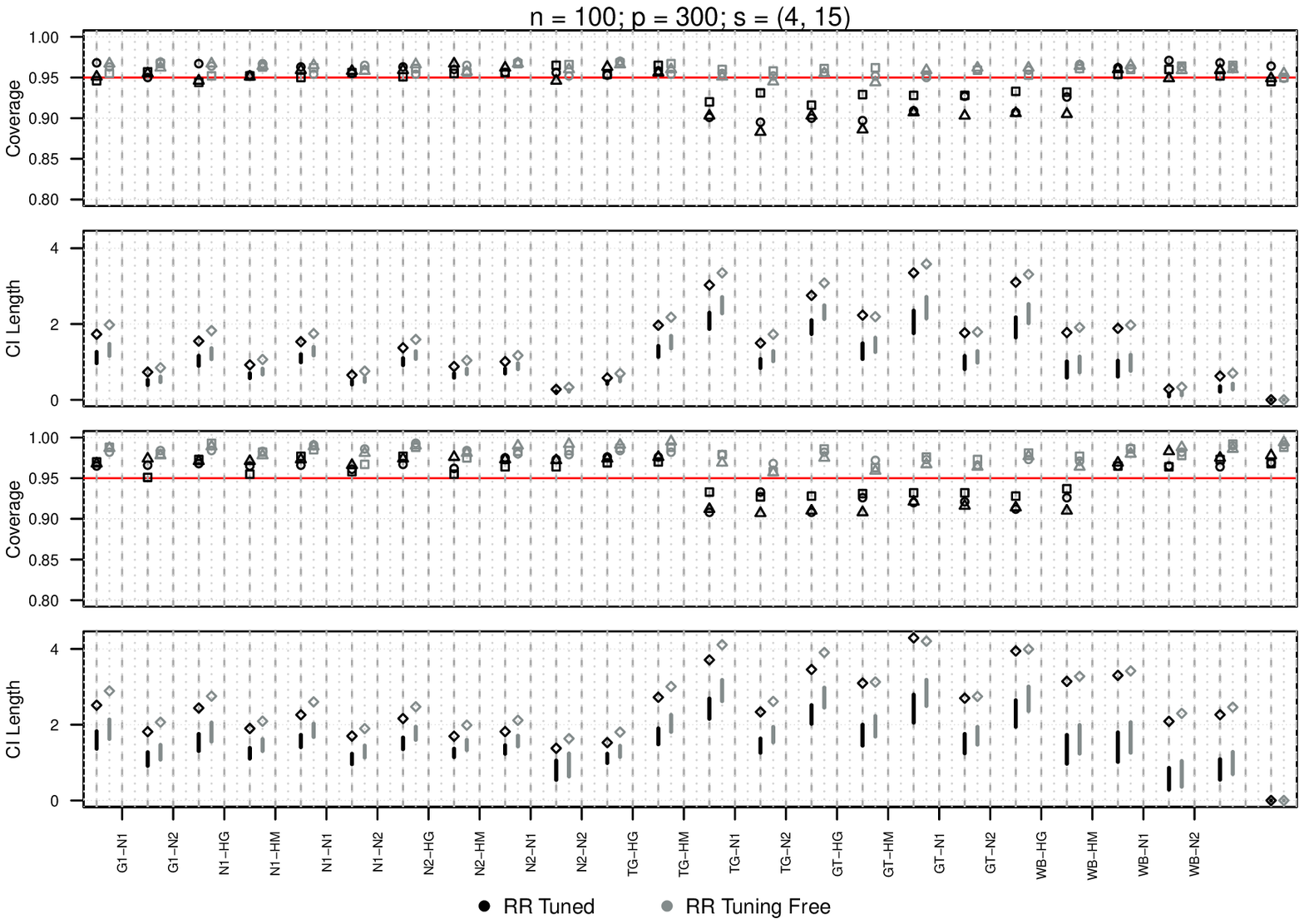}
\caption{Empirical coverage and confidence interval length for \emph{inactive variables} when $n = 100$ and $p = 300$ for \emph{exchangeable errors}. In the top two and bottom two panels, the true support of $\beta$ are $4$ and $15$ respectively. The first and third panels show empirical coverage rates for each procedure. In the bottom panel, the line segment indicates the $.25$ and $.75$ quantiles of the confidence interval lengths (averaged across all inactive variables for each run) and the single point indicates the $.99$ quantile. The labels on the horizontal axis indicate a different simulation setting and are coded as ``Covariate - Errors'' where the different covariate and error settings are detailed in the main text.}
\label{fig:rr_n_perm_med}
\end{figure}
\begin{figure}\centering
\includegraphics[scale=0.95]{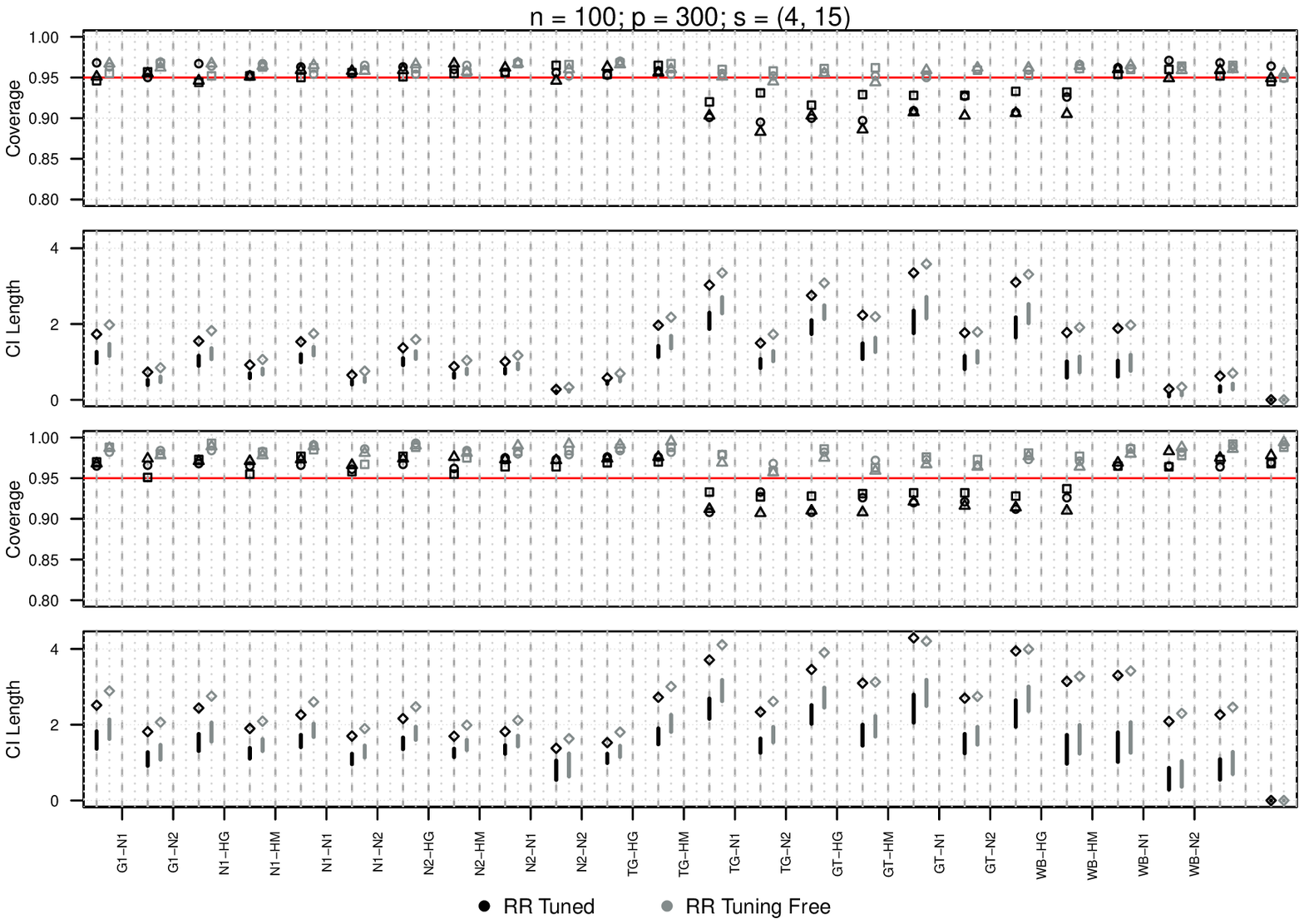}
\caption{Empirical coverage and confidence interval length for \emph{inactive variables} when $n = 100$ and $p = 300$ for \emph{sign symmetric errors}. All other elements remain the same as Figure~\ref{fig:rr_n_perm_med}.}
\label{fig:rr_n_sign_med}
\end{figure}

\FloatBarrier
\clearpage

\clearpage

\bibliography{icml_bib}
\bibliographystyle{icml2021}